\documentclass[11pt,a4paper]{article}
\usepackage[T1]{fontenc}
\usepackage{lmodern}
\usepackage[margin=28mm]{geometry}
\usepackage{amsmath,amssymb,amsthm,mathtools}
\usepackage{booktabs,longtable,array}
\usepackage{microtype}
\usepackage{needspace}
\usepackage[colorlinks=true,linkcolor=blue,citecolor=blue,urlcolor=blue]{hyperref}
\hypersetup{pdftitle={Finite images of braid group representations and algebraic solutions of KZ-type equations},
pdfauthor={Haru Negami},pdfkeywords={braid groups, middle convolution, unitarizability, finite images, KZ-type equations}}
\usepackage{enumitem}
\usepackage{fancyhdr}
\setlist{itemsep=3pt,topsep=5pt}
\numberwithin{equation}{section}
\newtheorem{theorem}{Theorem}[section]
\newtheorem*{maintheorem}{Main Theorem}
\newtheorem{proposition}[theorem]{Proposition}
\newtheorem{lemma}[theorem]{Lemma}
\newtheorem{corollary}[theorem]{Corollary}
\theoremstyle{definition}
\newtheorem{definition}[theorem]{Definition}
\newtheorem{example}[theorem]{Example}
\theoremstyle{remark}
\newtheorem{remark}[theorem]{Remark}
\DeclareMathOperator{\GL}{GL}
\DeclareMathOperator{\Spec}{Spec}

\DeclareMathOperator{\KLM}{KLM}

\newcommand{\RR}{\mathbb R}

\title{Finite images of braid group representations\\ and algebraic solutions of KZ-type equations}
\newcommand{\AuthorAffiliation}{Chiba University}
\newcommand{\AuthorAddress}{1-33 Yayoi-cho, Inage-ku, Chiba 263-8522, Japan}
\newcommand{\AuthorEmail}{negamiharu@gmail.com}
\newcommand{\PaperKeywords}{Braid groups; Long--Moody construction;
 middle convolution; unitarizability; finite images; KZ-type equations;
 Fuchsian systems; hypergeometric functions; algebraic solutions.}
\newcommand{\PaperMSC}{Primary 20F36; Secondary 20C15, 32G34, 34M35.}
\author{Haru Negami\\[0.5em]
 {\small\AuthorAffiliation}\\
 {\small\AuthorAddress}\\
 {\small\texttt{\AuthorEmail}}}
\date{19 September 2026}

\newcommand{\C}{\mathbb C}
\newcommand{\Q}{\mathbb Q}
\newcommand{\Z}{\mathbb Z}
\newcommand{\MCDR}{\operatorname{MC}^{\mathrm{DR}}}
\DeclareMathOperator{\rank}{rank}
\DeclareMathOperator{\lcm}{lcm}
\newcommand{\PG}{\mathbb P}
\begin{document}
\maketitle
\begin{abstract}
Finite monodromy provides a bridge between group representations and
algebraic solutions of differential equations. We study this connection
for the Katz--Long--Moody construction, which transforms representations
of the semidirect product of a free group and a braid group into new
representations of the same group and is related to
Knizhnik--Zamolodchikov (KZ)-type equations. For a fixed finite-image
input, we classify the parameter values for which the resulting
representations have finite image, both on the semidirect product and on
its free-group and braid-group subgroups. In particular, finiteness of
the braid-group image is independent of the admissible parameter. These
results give necessary and sufficient conditions for all solutions of
the corresponding regular-singular KZ-type equations to be algebraic.
On restriction to the free group, they also characterize finite monodromy
and algebraicity of all solutions of the associated Fuchsian systems,
connecting the classification to classical questions about algebraic
hypergeometric functions.
\end{abstract}

\noindent\textbf{Keywords.} \PaperKeywords

\smallskip
\noindent\textbf{2020 Mathematics Subject Classification.} \PaperMSC

\setcounter{secnumdepth}{3}
\section{Introduction}

Finite monodromy links the study of group representations to the
algebraicity of solutions of differential equations. In the theory of
special functions, determining which parameters give algebraic solutions
is a classical problem. For a regular-singular
algebraic connection, including regularity at infinity, finite monodromy
is equivalent to algebraicity of all local horizontal solutions
\cite{Deligne}. Thus a finite-monodromy criterion determines which
members of a parameterized family have entirely algebraic solutions.
Examples include the generalized hypergeometric equations studied by
Beukers--Heckman \cite{BH} and the Pochhammer equations studied by
Haraoka \cite{Haraoka1994}. In several variables, this question has
been studied for Lauricella $F_D$ \cite{Sasaki,CohenWolfart} and more
broadly for Appell--Lauricella and Horn functions \cite{Bod}.
For Lauricella $F_C$, Goto \cite{GotoFiniteFC} gives explicit parameter
conditions for finite irreducible monodromy and studies the structure
of the resulting finite groups. These studies connect the algebraicity
of special functions with concrete questions about linear group images.

For nonlinear special functions, Iwasaki \cite{IwasakiFiniteBranch}
relates algebraicity of Painlev\'e VI solutions to finite branching.
That finiteness concerns continuation orbits of individual solutions;
here we study images of linear representations.

A classical route to finite monodromy combines definiteness of invariant
Hermitian forms at every coefficient embedding with an integral model
\cite{RRV,GV,BH}. In several variables, related results include Beukers'
algebraicity criterion for irreducible $A$-hypergeometric systems
\cite{BeukersA} and Bod's classifications of algebraic
Appell--Lauricella and Horn hypergeometric functions \cite{Bod}. The Schwarz problem for
KZ equations also has earlier treatments \cite{StanevTodorov}.

The fixed-parameter arithmetic test has precedents in the
hypergeometric interlacing results of Beukers--Heckman \cite{BH},
Haraoka's criterion for rational generic Pochhammer systems
\cite[Theorem 1.2]{Haraoka1994}, and the rank-four Goursat analysis of
Radchenko--Rodriguez Villegas \cite[Sections 4.3 and 5]{RRV}.

A particularly direct antecedent is Ghate--Venkataramana
\cite[Theorem 8 and Sections 5--6]{GV}. For the specialized Gassner
representations considered there, with nontrivial local parameters and
nontrivial product, finiteness is characterized by definiteness at all
coefficient embeddings and by fractional-part conditions.
Their Lemma~11 gives the arithmetic finiteness principle for an
irreducible integral unitary group over a CM field.

From the representation-theoretic viewpoint, the Long--Moody construction
is an algebraic procedure that assigns a representation of $B_n$ to
a representation of the semidirect product
$\Gamma=F_n\rtimes B_n$ \cite{Long1994,BigelowTian}.
Here $B_n$ is the braid group on $n$ strands, $PB_n\subset B_n$ is
its pure braid subgroup, and $F_n=\langle x_1,\ldots,x_n\rangle$ is
the free group of rank $n$,
and the semidirect product is defined by the Artin action
$B_n\to\operatorname{Aut}(F_n)$; we follow the conventions of
\cite[Section~2.1]{KLM}, to which we refer for the explicit action
and the semidirect-product group law. Subsequent developments of the
Long--Moody construction include
Souli\'e's functorial constructions \cite{Soulie2019,Soulie2022} and
Takano's relation to twisted Alexander invariants \cite{Takano2024}.
In the usual application to a representation of $B_{n+1}$, one first
restricts to $F_n\rtimes B_n\subset B_{n+1}$ and obtains a
$B_n$-representation; successive applications decrease the number of
strands. Hiroe--Negami's Katz--Long--Moody (KLM) construction retains the
free-group action and produces a new representation of the same
semidirect product \cite{HN}. Thus it can be iterated at fixed $n$
to obtain successive representations of $F_n\rtimes B_n$.

Middle convolution was developed by Katz as a tool for rigid local
systems \cite{Katz}; Dettweiler--Reiter gave explicit matrix and
differential-system constructions \cite{DR2000,DR2007,DR2003}.
Its relation to the Long--Moody construction is part of the foundation
of KLM \cite{HN,KLM}.

For KLM, the link to special functions is provided by its analytic
interpretation. The Knizhnik--Zamolodchikov (KZ) equations originate in conformal field
theory \cite{KZ1984}, and KZ-type systems connect braid monodromy
with multivariable hypergeometric functions \cite{KLM,Haraoka2020}.
On the pure braid subgroup $PB_{n+1}\simeq F_n\rtimes PB_n$, the KLM
construction corresponds to Haraoka's multiplicative middle convolution
for KZ-type equations \cite{KLM,Haraoka2020}. We refer to this
identification as the \emph{Haraoka--Long correspondence}.
Consequently, finite image of the resulting pure braid representation
provides an algebraicity criterion for the corresponding KZ-type system.
Section~\ref{sec:algebraic} states this consequence with the monodromy
identification and regular-singularity assumptions explicit.

In this paper, we study the finite-image problem for the KLM
construction. This involves three related images: those of a normal free subgroup,
the braid subgroup, and the full semidirect product. We give finite-image
criteria and classify, for each fixed finite input, all convolution
parameters giving finite free or full output. In the finite-parameter
case, an explicit order bound makes the classification effective.
We also determine the exceptional inputs with infinitely many such
parameters and describe the different behavior of the braid image.
The results apply to algebraic solutions of the corresponding KZ-type systems.

\paragraph{Representations, unitarity, and invariant Hermitian forms.}
\begin{definition}[Representations]\label{def:representations}
By a finite-dimensional complex representation we mean the data
$(G,V,\rho)$ of a group $G$, a finite-dimensional complex vector space
$V$, and a homomorphism
\[
 \rho\colon G\longrightarrow\GL(V).
\]
When $G$ and $V$ are understood, we denote the representation simply
by $\rho$. Its \emph{image} is the subgroup $\rho(G)\subseteq\GL(V)$;
in particular, ``finite image'' means that this subgroup is finite.
\end{definition}

\begin{definition}[Unitarity and invariant Hermitian forms]\label{def:unitarity}
Following
\cite[Definition~7]{KLM} and the convention of
\cite[Definition~2.2]{Long1994}, we say
that $\rho$ is \emph{unitary relative to $H$} if $H$ is a nondegenerate
Hermitian matrix and
\[
 \rho(g)^*H\rho(g)=H\qquad(g\in G).
\]
We call a representation satisfying this identity for some
nondegenerate $H$, which may be indefinite, \emph{generally
unitary}. The associated sesquilinear form
\[
 h_H(v,w)=v^*Hw
\]
is the \emph{invariant Hermitian form}; its associated quadratic form
is $q_H(v)=v^*Hv$. If $H$ can be chosen positive definite, we call
$\rho$ \emph{unitary}, or equivalently \emph{positively unitarizable}.
After choosing an $H$-orthonormal basis, its matrices lie in the usual
unitary group. Thus every unitary representation is generally unitary,
whereas a generally unitary representation preserving only an
indefinite form need not be unitary in this sense.
\end{definition}

For a fixed input $\rho:\Gamma\to\GL(V)$, write
$T_\lambda=\KLM_\lambda(\rho)$, $\lambda\in\C\setminus\{0,1\}$.
Restricting the KLM output to $F_n$ recovers the Dettweiler--Reiter
multiplicative middle convolution of the restricted input
\cite{DR2000,DR2007}:
$T_\lambda|_{F_n}=\MCDR_\lambda(\rho|_{F_n})$ in our conventions.
The KLM output also carries the compatible braid-group action supplied
by the input representation of the semidirect product.
For the detailed definitions and this relation, see
\cite[Section~2.3]{KLM}.
We distinguish the restrictions to $F_n$ and $B_n$ from the
representation of the full semidirect product $\Gamma$.
Invariant Hermitian forms provide the main tool for our finite-image
criterion. For
$|\lambda|=1$, $\lambda\ne1$, KLM carries an input preserving a
nondegenerate Hermitian form to an output preserving a canonical
nondegenerate Hermitian form, whose signature may be indefinite
\cite{KLM}. For input preserving a positive definite form, the closed signature formula
of \cite{Signature} determines exactly when this canonical output form
is definite, and hence gives positively unitarizable output. For
reducible output, positive unitarizability can also arise by changing
the signs on invariant components. This distinction is part of the
structural analysis needed to determine finiteness of the image.

The meaning of unitarity is significant here. Long's
Definition~2.2 and Theorem~2.8 use a nondegenerate Hermitian form,
which may be indefinite; the preservation theorem there is stated
for generic unit-modulus parameters \cite{Long1994}. Theorems~8--9 of \cite{KLM}
likewise give invariance and nondegeneracy. In contrast,
Theorem~\ref{thm:unitary-lift} characterizes \emph{positive}
unitarizability of the full KLM output by its free restriction.
Its additional step is the braid-compatible sign correction on
convolved input isotypic components, including multiplicities.
The input may be reducible and need not have finite image.
The arbitrary-rank parabolic-cohomology signature theorem
\cite[Theorem 2.3]{DW}, applied to
$(g_1,\ldots,g_n,\lambda I_N,(\lambda P_n)^{-1})$, gives the reversed
ordered signature used here; the numerical comparison in
Remark~\ref{rem:parabolic-comparison} does not identify the pairings themselves.
Lemma~\ref{lem:forms} gives direct proofs of free-generator invariance,
the radical equality, and compatibility with orthogonal decompositions.
Full KLM invariance is recalled separately in Lemma~\ref{lem:full-invariance}.

The preceding connections lead to two questions:
\begin{enumerate}[label=(\arabic*)]
\item For finite input, when is the KLM output finite, and can all
parameters with finite free or full image be determined effectively?
\item How does finiteness of the braid restriction relate to that of
the free and full images, and how do its ordinary and projective
images change at exceptional parameters?
\end{enumerate}
We answer these questions under the input assumptions stated below.
The proofs combine positive unitarizability, arithmetic finiteness,
a count of conjugates giving effective parameter bounds, and equivariant
splittings. The all-parameter classification requires this additional
counting argument: a finite-image test at each root of unity by itself
does not bound the orders that must be examined. The examples include a rank-three non-rigid
finite input with finite output, as well as a classical finite Burau
specialization illustrating the dimension drop at an exceptional parameter.
For the $\mathfrak S_3$, scalar, and $G(3,1,3)$ inputs, we give explicit
integer-arithmetic procedures that verify the complete parameter
classifications.

\paragraph{Finite-image criteria and parameter classification.}
Theorem~\ref{thm:free} applies the arithmetic principle constituent
by constituent to a specified finite input, using middle-convolution
irreducibility \cite[Theorem 2.4]{DR2007}; see also
\cite[Theorem 2.4(iii)]{DR2003}.
Lemma~\ref{lem:simple} handles its exceptional cases separately.
Its full-image consequence additionally uses
Theorem~\ref{thm:unitary-lift}; finiteness of $\rho(F_n)$ alone does not
suffice. Theorem~\ref{thm:parameters} addresses the additional
classification problem with fixed input and initially unbounded
parameter order. The sieve using conjugates over the input cyclotomic field supplies
the effective order bound needed to turn the fixed-parameter test into
a complete finite enumeration. Together with the exceptional-case
analysis, it classifies all parameter sets without a rigidity assumption.
The order bound uses the structural properties of the canonical form;
the closed signature formula evaluates the resulting finite tests from
spectral data. Section~\ref{sec:parameters} makes this distinction
explicit.
Belkale's finite classification \cite[Corollary 1.3.4]{Belkale}
instead fixes a common multiple of the local orders and the number
of punctures, within the class of irreducible rigid local systems.
The present classification varies the convolution parameter for a
specified finite input, rather than varying all local systems.
General algorithms for deciding finiteness of a specified matrix group
already exist \cite{DFO}; Theorem~\ref{thm:braid-dichotomy} reduces
braid-image finiteness for the parameter family to finiteness of the
single representation $T_{\rm gen}$.

\paragraph{Exceptional parameters and braid actions.}
Long's original-representation summand \cite[Theorem 2.11]{Long1994},
the reduced construction \cite[Section 6]{BigelowTian}, and the
equivariant maps used in the construction of \cite{HN} precede this work.
The nonexceptional braid representation is already parameter-independent
in \cite[Corollary 1.2 and Proposition 6.11]{Signature}.
The additional assertion here is that finiteness persists across
quotients at exceptional parameters when the input braid image is finite.
Section~\ref{sec:resonance} uses these equivariant maps after quotienting by the
local fixed subspaces, isolates each product eigenspace $V_\mu$, and
constructs an equivariant retraction. Finite input braid image permits
averaging even when the product matrix is not semisimple. We derive
surjective homomorphisms of ordinary images with finite kernel, the all-parameter
braid-finiteness alternative, and the corresponding projective kernels.
The two projective exact sequences are elementary group-theoretic
consequences of the split representation; they specify precisely
which additional relations arise in the projective image at exceptional parameters. A direct sum
of representations does not assert that the image-group extension
splits.

There is also recent work on finite braid or mapping-class-group
\emph{orbits} of local systems: Lam--Landesman--Litt study rank-two
character varieties using middle convolution and finite reflection
groups \cite{LLL}, and Vayalinkal gives explicit enumeration methods
\cite{Vayalinkal}. Those orbits consist of conjugacy classes of free-group
representations. Their finiteness is different from finiteness of the
linear image $T_\lambda(B_n)$ studied here; the former setting can
include representations with Zariski-dense image on $F_n$.

\subsection{Conventions}\label{sec:conventions}

\paragraph{Groups and representations.}
Throughout, $n\ge2$. The standard generators of $B_n$ are
$\sigma_1,\ldots,\sigma_{n-1}$. Its action on $F_n$ defining
$\Gamma=F_n\rtimes B_n$ is
\begin{equation}\label{eq:artin-action}
\begin{aligned}
 \sigma_i x_i\sigma_i^{-1}&=x_{i+1},\\
 \sigma_i x_{i+1}\sigma_i^{-1}&=x_{i+1}^{-1}x_ix_{i+1},\\
 \sigma_i x_j\sigma_i^{-1}&=x_j\qquad(j\notin\{i,i+1\}),
\end{aligned}
\end{equation}
for $1\le i<n$ and $1\le j\le n$. We write $\mathfrak S_n$ for
the symmetric group, so $PB_n=\ker(B_n\to\mathfrak S_n)$.
Write $I_d$ for the $d\times d$ identity matrix and $I_W$ for the
identity map on a vector space $W$. For an input space $V$, put
$N:=\dim V$; thus the identity matrices on $V$ and $V^{\oplus n}$
are $I_N$ and $I_{nN}$, respectively.
For an input representation $\rho$ of $F_n$ or $\Gamma$, put
$g_j=\rho(x_j)$ ($1\le j\le n$) and $P_n=g_1\cdots g_n$.
When $\rho$ is defined on $\Gamma$, also put $s_i=\rho(\sigma_i)$
($1\le i<n$). An input matrix $g_j$ is called \emph{active} if $g_j\ne I_N$;
on an invariant subspace $U\subseteq V$, it is called active if
$g_j|_U\ne I_U$. The \emph{output} is
$\MCDR_\lambda(\rho)$ for an $F_n$-input, or
$T_\lambda=\KLM_\lambda(\rho)$ for a $\Gamma$-input.
The adjectives \emph{free}, \emph{braid}, and \emph{full} refer to
restriction to $F_n$, restriction to $B_n$, and the representation of
$\Gamma$, respectively. Thus finite full input means
$|\rho(\Gamma)|<\infty$. The full output is $T_\lambda$, and its
free restriction is $\MCDR_\lambda(\rho|_{F_n})$.

\paragraph{Forms and spectral notation.}
Unitarity terminology is as fixed above;
\emph{unitarizable} without a modifier means positively unitarizable.
All representations act on column vectors on the left, and products of
group elements are sent to products of matrices in the written order.
We write $\operatorname{sig}\mathcal H=(p,q)$ for the numbers
of positive and negative eigenvalues of a Hermitian matrix representing
$\mathcal H$, in that order;
zero eigenvalues are omitted. Reversing the signature means replacing
$(p,q)$ by $(q,p)$. A form is \emph{definite} if it is positive or
negative definite; \emph{positive} and \emph{negative} forms mean
positive and negative definite forms. We write $\Spec(A)$ for the
set of eigenvalues of $A$. For unit-circle eigenvalues, the
\emph{eigenangles} are their arguments divided by $2\pi$, chosen in
$[0,1)$ and listed with algebraic multiplicity.
For real $x$, write $\{x\}=x-\lfloor x\rfloor$ for its fractional part.

\paragraph{Arithmetic conventions.}
A finite-image representation
can be realized over a number field after a change of basis, by
Brauer's splitting-field theorem \cite[p.~461]{Brauer1945} and complete
reducibility of complex representations of finite groups. For effective
statements, an exact number-field presentation and the input matrices are
part of the data. A \emph{coefficient field} $E\subset\C$ is a number
field containing the entries of the matrices under consideration;
we enlarge it when additional algebraic parameters are introduced.
A \emph{coefficient embedding} is a field embedding
$\tau:E\hookrightarrow\C$, applied entrywise to those matrices.
For a number field $E$, its ring of integers
$\mathcal O_E$ consists of the elements of $E$ that satisfy a monic
polynomial with coefficients in $\Z$. In particular, any root of unity
$\lambda\in E$ belongs to $\mathcal O_E^\times$: both $\lambda$ and
$\lambda^{-1}$ satisfy $X^r-1=0$ for some integer $r\ge1$.
The zero-dimensional representation has trivial image.

\subsection{Main results}\label{sec:main-results}

The results of this paper fall into four parts:
\begin{enumerate}[label=(\arabic*)]
\item positive unitarizability and finiteness;
\item complete classification of finite-image parameters $\lambda$;
\item algebraic solutions of KZ-type equations;
\item exceptional parameters and parameter-independent braid finiteness.
\end{enumerate}
\subsubsection{Positive unitarizability and finiteness}\label{sec:main-unitarity}

If $\rho:\Gamma\to\GL(V)$ is positively unitarizable and
$|\lambda|=1$, $\lambda\ne1$, then
\begin{equation}\label{eq:intro-unitary}
 T_\lambda|_{F_n}\text{ is positively unitarizable}
 \iff T_\lambda\text{ is positively unitarizable}.
\end{equation}
The input may be reducible and have infinite image.
For finite-image $F_n$-input and a root of unity $\lambda\ne1$,
the output has finite image if and only if the canonical form on each
nonzero convolved irreducible constituent is definite at every
coefficient embedding. For finite full input, moreover,
\begin{equation}\label{eq:intro-full}
 T_\lambda(F_n)\text{ is finite}
 \iff T_\lambda(\Gamma)\text{ is finite}
 \qquad(\lambda\ne0,1),
\end{equation}
and these conditions are also equivalent to finiteness of the
corresponding projective images.

\subsubsection{Complete classification of finite-image parameters \texorpdfstring{$\lambda$}{lambda}}\label{sec:main-parameters}

\begin{maintheorem}[Classification of finite-output parameters]
Let $\rho:F_n\to\GL(V)$ have finite image, and define
\[
 \mathcal F(\rho)=\{\lambda\in\C\setminus\{0,1\}:
       \MCDR_\lambda(\rho)(F_n)\text{ is finite}\}.
\]
Exactly one of the following holds.
\begin{enumerate}[label=(\roman*)]
\item Some irreducible constituent has at least two active generators.
Then $\mathcal F(\rho)$ is finite and effectively enumerable from
exact input data.
\item Every irreducible constituent has at most one active generator,
and $\rho$ is nontrivial. Then $\mathcal F(\rho)$ consists of all
roots of unity other than $1$.
\item The input is trivial. Then the output is zero and
$\mathcal F(\rho)=\C\setminus\{0,1\}$.
\end{enumerate}
\end{maintheorem}

\subsubsection{Algebraic solutions of KZ-type equations}\label{sec:main-algebraicity}

Suppose $\rho$ has finite full image and $\nabla_\lambda$ is a flat
algebraic connection on the ordered configuration space
$X_{n+1}=\{(z_0,\ldots,z_n)\in\C^{n+1}:z_i\ne z_j\ (i\ne j)\}$,
regular singular also at infinity, whose monodromy is isomorphic to
$T_\lambda|_{F_n\rtimes PB_n}$. Then all its horizontal solutions are
algebraic if and only if $T_\lambda(F_n)$ is finite. Thus the finite-image
parameter classification also determines the algebraicity locus of any
such family within its parameter domain.

\subsubsection{Exceptional parameters and parameter-independent braid finiteness}\label{sec:main-resonance}

If $\rho(B_n)$ is finite, then
\[
 \{\lambda\in\C\setminus\{0,1\}:T_\lambda(B_n)\text{ is finite}\}
 =\C\setminus\{0,1\}\quad\text{or}\quad\varnothing.
\]
At an exceptional parameter, the generic braid representation splits
as the sum of the specialized output and a finite-image input summand.
The resulting surjection of braid images has finite kernel. Hence
finite braid image for every parameter can coexist with a finite set
of parameters giving finite full image.

\medskip
\noindent
The paper is organized as follows.
Section~\ref{sec:free} develops the preliminary finite-image criterion
for the middle convolution of a finite-image free-group representation.
Section~\ref{sec:full} proves that, for finite full input, finiteness of
the KLM representation of $F_n\rtimes B_n$ is determined by its
restriction to $F_n$, using a lifting theorem for positive unitarizability.
Section~\ref{sec:parameters} classifies the parameters giving finite
image for a fixed finite input.
Section~\ref{sec:algebraic} applies these criteria to algebraic solutions
of regular-singular systems, including KZ-type equations under the
stated monodromy identification.
Section~\ref{sec:resonance} studies exceptional parameters and proves that finiteness
of the braid output is parameter-independent when the input braid
image is finite.
Section~\ref{sec:discussion} discusses further directions.
Appendix~\ref{sec:examples} collects the examples in the order of the
corresponding results in the main text, including changes in braid images
and the necessity of the hypotheses.

\section{Middle convolution, invariant forms, and representations of \texorpdfstring{$F_n$}{Fn}}\label{sec:free}

\subsection{Construction and structural inputs}
\begin{definition}[Dettweiler--Reiter convolution and middle convolution]\label{def:dr-convolution}
Initially let $\rho:F_n\to\GL(V)$ be a finite-dimensional complex
representation, put $N=\dim_\C V$, write $g_j=\rho(x_j)$
for $1\le j\le n$, and let $\lambda\in\C\setminus\{0,1\}$.
No finiteness assumption is needed to define the construction.
On $V^{\oplus n}$ let $A_i$ be the identity outside block row
$i$, whose entries are
\begin{equation}\label{eq:Ai}
 \bigl(\lambda(g_1-I_N),\ldots,\lambda(g_{i-1}-I_N),
       \lambda g_i,g_{i+1}-I_N,\ldots,g_n-I_N\bigr).
\end{equation}
Write
\[
 K=\bigoplus_{j=1}^n\ker(g_j-I_N),\quad
 L=L_\lambda=\bigcap_{j=1}^n\ker(A_j-I_{nN}),\quad
 Q_\lambda=V^{\oplus n}/(K+L).
\]
The \emph{Dettweiler--Reiter convolution} is the representation
$C_\lambda(\rho):F_n\to\GL(V^{\oplus n})$ defined by
$x_i\mapsto A_i$ for $1\le i\le n$.
The subspace $K+L$ is invariant under these operators. The induced
representation on $Q_\lambda$, given by
$x_i\cdot[v]=[A_i v]$, is the \emph{Dettweiler--Reiter multiplicative
middle convolution} $\MCDR_\lambda(\rho)$ \cite{DR2000,DR2007}.
The summand $\ker(g_j-I_N)$ in $K$ is placed in coordinate $j$.
\end{definition}

\begin{lemma}[Irreducibility under multiplicative middle convolution]\label{lem:simple}
If $U$ is a finite-dimensional irreducible complex representation of
$F_n$ and $\lambda\ne0,1$, then $\MCDR_\lambda(U)$ is either zero or irreducible.
\end{lemma}
\begin{proof}
For an irreducible tuple with at least two nonidentity generators,
conditions $(*)$ and $(**)$ preceding \cite[Theorem 2.4]{DR2003}
hold: they exclude one-dimensional submodules and quotients on which
at most one generator acts nontrivially. Part~(iii) then gives
irreducibility; see also \cite[Theorem 2.4]{DR2007}.
If at most one generator is nonidentity, the
image is cyclic, so an irreducible complex input has dimension one.
For the trivial input, $K=U^{\oplus n}$ and the output is zero.
If only $g_j=a\ne1$ is nonidentity, $U^{\oplus n}/K$ is one-dimensional.
Its generator $x_j$ acts by $\lambda a$, with all other generators
acting trivially. The additional quotient by $L$ is zero when
$\lambda a=1$, and otherwise leaves this one-dimensional module.
\end{proof}

For unitarizable input and $|\lambda|=1$, $\lambda\ne1$,
the radical equality for the canonical form was established in \cite{KLM}.
We recall a direct verification of this equality and free-generator
invariance in Lemma~\ref{lem:forms}(i) below.
For numerical evaluation we additionally use the closed signature
formula of \cite{Signature}. The arbitrary-rank parabolic-cohomology
signature in \cite[Theorem 2.3]{DW} is a closely related earlier result.
The numerical comparison, including the reversal of signs, is explained
in Remark~\ref{rem:parabolic-comparison} in the discussion.

\begin{definition}[Canonical intermediate Hermitian form and radical]\label{def:ambient-form}
Suppose that the input
preserves a nondegenerate Hermitian form represented by a Hermitian
matrix $H$, assume
$|\lambda|=1$, $\lambda\ne1$, and choose a square root
$\lambda^{1/2}$. On the intermediate convolution space
$V^{\oplus n}$, before taking the quotient by $K+L_\lambda$, the form $\widetilde H_\lambda$ introduced in
\cite[Section~5.1]{KLM} has block entries
\begin{equation}\label{eq:ambient-form}
 (\widetilde H_\lambda)_{jk}=
 \begin{cases}
 \lambda^{-1/2}H(g_j^{-1}-I_N)(g_k-I_N),&j<k,\\
 \lambda^{-1/2}H(g_j^{-1}-\lambda I_N)(g_j-I_N),&j=k,\\
 \lambda^{1/2}H(g_j^{-1}-I_N)(g_k-I_N),&j>k.
\end{cases}
\end{equation}
The radical of this form is
\[
 \operatorname{rad}(\widetilde H_\lambda)
 :=\{v\in V^{\oplus n}\mid
       \widetilde H_\lambda(v,w)=0
       \text{ for every }w\in V^{\oplus n}\}.
\]
Since the space is finite-dimensional and \eqref{eq:ambient-form}
specifies the Hermitian matrix representing the form, this radical
is precisely the kernel of that matrix:
\[
 \operatorname{rad}(\widetilde H_\lambda)
 =\ker\widetilde H_\lambda.
\]
\end{definition}
The equality $\ker\widetilde H_\lambda=K+L$ was established in
\cite{KLM}. For completeness, Lemma~\ref{lem:forms}(i) below recalls
a direct proof requiring only nondegeneracy of the input form.
Consequently, $\widetilde H_\lambda$ induces the canonical
nondegenerate form on the middle-convolution quotient $Q_\lambda$.

\begin{definition}[Intermediate Hermitian form in an orthonormal basis]\label{def:ambient-pencil}
For input preserving a positive definite form, choose an orthonormal basis, so that $H=I_N$. Write
$\lambda=e^{2\pi i l}$ with $0<l<1$, choose
$\lambda^{1/2}=e^{\pi i l}$. Put $N:=\dim_\C V$ and define
\begin{equation}\label{eq:D-definition}
 \begin{aligned}
 D &: V^{\oplus n}\longrightarrow V,\qquad
 (v_1,\ldots,v_n)\longmapsto\sum_{j=1}^n(g_j-I_N)v_j,\\
 D &= (g_1-I_N\ \cdots\ g_n-I_N)\in M_{N\times nN}(\C).
 \end{aligned}
\end{equation}
Here $v_j\in V$ for $1\le j\le n$. The definition of $D$ is
independent of $\lambda$. Then \eqref{eq:ambient-form} can be written
\begin{equation}\label{eq:canonical-pencil}
 \widetilde H_\lambda=\mathcal H(l)
 =\sin(\pi l)B+\cos(\pi l)D^*D,
\end{equation}
where
\[
 B_{jj}=-i(g_j-g_j^*),\qquad
 B_{jk}=-i(g_j-I_N)^*(g_k-I_N)\ (j<k),\qquad B_{kj}=B_{jk}^*.
\]
The expression for $\mathcal H(l)$ in \eqref{eq:canonical-pencil}
extends continuously to $l=0,1$.
\end{definition}
We collect precisely the form properties used below.

\begin{definition}[Canonical quotient form]\label{def:quotient-form}
With the assumptions and notation of Definition~\ref{def:ambient-form},
the \emph{canonical quotient form} on $Q_\lambda$ is defined by
\begin{equation}\label{eq:canonical-quotient-form}
 \mathcal H_{\lambda,h}([v],[w]):=v^*\widetilde H_\lambda w,
 \qquad [v],[w]\in Q_\lambda.
\end{equation}
Here $v,w\in V^{\oplus n}$ are representatives, and $h$ is the input
form represented by the Hermitian matrix $H$. Lemma~\ref{lem:forms}(i) below proves that
this is independent of the representatives and is nondegenerate.
\end{definition}

\begin{lemma}[Canonical free-group form and functorial decompositions]\label{lem:forms}
Let $\rho:F_n\to\GL(V)$ preserve a nondegenerate Hermitian form
$h$, write $g_j=\rho(x_j)$ for $1\le j\le n$, and assume
$|\lambda|=1$, $\lambda\ne1$. Let $H$ be the Hermitian matrix representing $h$ in the chosen basis,
and use the intermediate matrix $\widetilde H_\lambda$ of
\eqref{eq:ambient-form}, with the square-root convention of
\eqref{eq:canonical-pencil}.
\begin{enumerate}[label=(\roman*)]
\item The radical of $\widetilde H_\lambda$ is $K+L$ \cite{KLM}. Hence the
canonical quotient form of Definition~\ref{def:quotient-form} is
well-defined and nondegenerate. It is
invariant under $\MCDR_\lambda(\rho)$, namely
\[
 \mathcal H_{\lambda,h}([A_i v],[A_i w])
 =\mathcal H_{\lambda,h}([v],[w])\qquad(1\le i\le n).
\]
\item Assume in addition that $h$ is positive definite. Suppose
$V=\bigoplus_{a=1}^s V_a$ is an $h$-orthogonal
sum of $F_n$-invariant subspaces. Put
$\rho_a:=\rho|_{V_a}$ and $h_a:=h|_{V_a\times V_a}$.
Then the induced identification
\[
 \MCDR_\lambda(\rho)\simeq
 \bigoplus_{a=1}^s\MCDR_\lambda(\rho_a)
\]
identifies $\mathcal H_{\lambda,h}$ with the orthogonal sum of the
forms $\mathcal H_{\lambda,h_a}$ defined by
\eqref{eq:canonical-quotient-form} for these restricted inputs.
In particular, suppose an input summand is $U\otimes \mathcal M$, where
$U$ is an $F_n$-module and $\mathcal M$ is a finite-dimensional multiplicity
space on which $F_n$ acts trivially. If
$g_j|_{U\otimes \mathcal M}=g_{j,U}\otimes I_{\mathcal M}$ for $1\le j\le n$ and
the input form on this summand is $h_U\otimes k$, with $h_U$ and
$k$ positive definite, then its output is
$\MCDR_\lambda(U)\otimes \mathcal M$ with form $\mathcal H_U\otimes k$.
Here $I_{\mathcal M}$ is the identity on $\mathcal M$, and
$\mathcal H_U:=\mathcal H_{\lambda,h_U}$ is the canonical quotient
form on $\MCDR_\lambda(U)$ defined by
\eqref{eq:canonical-quotient-form} using the tuple
$(g_{1,U},\ldots,g_{n,U})$.
\end{enumerate}
No finite-image assumption is imposed.
\end{lemma}
\begin{proof}
We prove (i) directly from \eqref{eq:Ai} and
\eqref{eq:ambient-form}; these are the invariance and radical properties
appearing in \cite[Theorem 8]{KLM} and
\cite[Theorem 1.1]{Signature}.
Put $N=\dim_\C V$.
For $1\le j\le n$, let $\iota_j:V\to V^{\oplus n}$ be the
inclusion into coordinate $j$ and define the linear map
\[
 e_j:V^{\oplus n}\longrightarrow V,\qquad
 e_j(v)=\lambda\sum_{k<j}(g_k-I_N)v_k
       +(\lambda g_j-I_N)v_j+\sum_{k>j}(g_k-I_N)v_k.
\]
Thus $A_j-I_{nN}=\iota_j e_j$, and
$L=\bigcap_{j=1}^n\ker e_j$. We also use $e_j$ for the matrix of
this map in the fixed bases. The relations $g_j^*H=Hg_j^{-1}$ and
$\overline{\lambda^{1/2}}=\lambda^{-1/2}$ show that
$\widetilde H_\lambda$ is Hermitian. Its block row $j$ satisfies
\begin{equation}\label{eq:ambient-row-factor}
 (\widetilde H_\lambda v)_j
 =\lambda^{-1/2}H(g_j^{-1}-I_N)e_j(v).
\end{equation}
Indeed, the diagonal block uses
$(g_j^{-1}-\lambda I_N)(g_j-I_N)
=(g_j^{-1}-I_N)(\lambda g_j-I_N)$.

\emph{Radical.} Since $H$ is invertible, \eqref{eq:ambient-row-factor}
gives $v\in\ker\widetilde H_\lambda$ if and only if
$e_j(v)\in\ker(g_j-I_N)$ for every $1\le j\le n$.
For such a $v$, put $k_j=(\lambda-1)^{-1}e_j(v)$ and
$k=(k_1,\ldots,k_n)$. Then $k\in K$ and
$e_j(k)=(\lambda-1)k_j=e_j(v)$, so $v-k\in L$.
Conversely, for any $k\in K$ and $u\in L$ one has
$e_j(k+u)=(\lambda-1)k_j\in\ker(g_j-I_N)$.
Thus $\ker\widetilde H_\lambda=K+L$.

\emph{Invariance.} Expanding $A_i=I_{nN}+\iota_i e_i$ and using
\eqref{eq:ambient-row-factor} and its adjoint gives, for $1\le i\le n$,
\[
\begin{aligned}
 A_i^*\widetilde H_\lambda A_i-\widetilde H_\lambda
 &=\lambda^{-1/2}e_i^*H
 \bigl[(g_i^{-1}-I_N)+\lambda(g_i-I_N)\\
 &\hspace{38mm}+(g_i^{-1}-\lambda I_N)(g_i-I_N)\bigr]e_i=0.
\end{aligned}
\]
The bracket vanishes by expansion. The radical equality and invariance
therefore give the well-defined, nondegenerate invariant quotient form
in (i). This argument requires only a nondegenerate invariant input
form, not positive definiteness or finite image.

For (ii), the displayed blocks preserve orthogonal direct sums and tensor with
$k$ on a multiplicity space. The subspaces $K,L$ decompose in the same
way, since they are formed from kernels of these block operators.
Passing to the quotient proves (ii). A change to an orthonormal input
basis is used for these form identities, not to alter a number-field
lattice model. The assertions concern a nondegenerate invariant form,
which is not asserted to be positive.
\end{proof}

\subsection{Arithmetic criterion for finite input}\label{sec:free-arithmetic}
For the remainder of this section, assume that $\rho(F_n)$ is finite.
Put $G:=\rho(F_n)$ and let $M$ be its exponent, the least common
multiple of the orders of its elements. The exponent form of Brauer's
splitting-field theorem (see Isaacs \cite[Theorem 10.3]{Isaacs}) implies that every irreducible
complex representation of $G$ is realizable over
$\mathbb Q(\zeta_M)$, where $\zeta_M=e^{2\pi i/M}$.
By complete reducibility, the same holds for the given representation.
We may therefore choose a number field $E\subset\C$ over which the
input is defined. The cyclotomic field above is one possible choice;
the arguments below apply to any such coefficient field. For the
chosen field there are a finite-dimensional $E$-vector space
$V_E$ and a representation $\rho_E:F_n\to\GL_E(V_E)$ whose extension
of scalars to $\C$ is isomorphic to the original representation.
This change of basis does not change finiteness of the input or output
image. Here $\GL_E(V_E)$ denotes the group of invertible $E$-linear
maps on $V_E$.

In this arithmetic discussion we write $V$ and $\rho$ for $V_E$ and
$\rho_E$, respectively. For a nonzero complex parameter $\lambda\in\C^\times$,
the convolution is initially taken on $V\otimes_E\C$; choosing a
number-field model of the input does not restrict $\lambda$ to be
algebraic. We first show that, apart from the trivial input, parameters
which are not roots of unity cannot give a finite output image.

\begin{proposition}[Parameters which are not roots of unity]\label{prop:nonroot}
If the finite input is nontrivial and $\lambda\in\C^\times$ is not a root
of unity, then $\MCDR_\lambda(\rho)(F_n)$ is infinite. If every $g_j=I_N$,
the output is zero for every nonzero parameter.
\end{proposition}
\begin{proof}
Put $P_n=g_1\cdots g_n$ and
$F(v)=(g_2\cdots g_nv,\ldots,g_nv,v)$. For $\lambda\ne1$, direct
subtraction of consecutive nonzero rows of $A_i-I_{nN}$ gives
\[
 L=F(\ker(\lambda P_n-I_N)),\qquad K\cap L=0.
\]
Since $P_n$ has finite order, $L=0$ for a parameter which is not a root
of unity. Choose $g_i$ with an eigenvalue $\mu\ne1$; such an eigenvalue
exists because the nonidentity finite-order matrix is semisimple.
On a vector in the $i$th coordinate with eigenvalue $\mu$, $A_i$ has
eigenvalue $\lambda\mu$. This vector has nonzero image modulo $K$,
so the eigenvalue survives the quotient. Since $\lambda\mu$ is not a
root of unity, the induced $A_i$ has infinite order, and the output
image is infinite.
\end{proof}

Consequently, a finite output image for a nontrivial finite input can
occur only when $\lambda$ is a root of unity. In that case we replace
$E$ by $E(\lambda)$ and continue to denote the enlarged field by $E$.
The input space and representation are extended to this enlarged
field as well. The matrices $A_i$, the subspaces $K,L$, and the
quotient $Q_\lambda$ are then all defined over $E$; for the remainder of this section,
$Q_\lambda$ denotes this $E$-vector space. Its extension of scalars
to $\C$ is the complex output considered above.
Here $\mathcal O_E$ is the ring of integers of $E$. Since $\lambda$
is a root of unity, $\lambda\in\mathcal O_E^\times$.

We use the classical arithmetic finiteness argument: an invariant
full $\mathcal O_E$-lattice, together with positive definite invariant
forms at every embedding of $E$ into $\C$, forces the image to be
finite. See \cite[Sections 4.3 and 5]{RRV} for this argument in the
setting of Goursat local systems, and \cite[Lemma 11]{GV} for an
irreducible integral unitary group over a CM field.
Lemma~\ref{lem:lattice} supplies
the lattice; the proof of Theorem~\ref{thm:free} recalls how restriction
of scalars reduces finiteness to the intersection of a compact group
with the discrete group of lattice automorphisms.

\begin{definition}[Full invariant lattice]\label{def:lattice}
A \emph{full $\mathcal O_E$-lattice} in $Q_\lambda$ means a
finitely generated $\mathcal O_E$-submodule $\Lambda_Q\subset Q_\lambda$
such that $E\Lambda_Q=Q_\lambda$. It is invariant under the output if
every output group element preserves $\Lambda_Q$.
\end{definition}

For an embedding $\tau:E\hookrightarrow\C$, put
$V^\tau=V\otimes_{E,\tau}\C$ and
$Q_\lambda^\tau=Q_\lambda\otimes_{E,\tau}\C$.
The representation $\rho^\tau$ on $V^\tau$ is obtained by applying
$\tau$ to the input matrix entries in an $E$-basis; it is called the
\emph{conjugate input}. The corresponding \emph{conjugate output} is
$\MCDR_{\tau(\lambda)}(\rho^\tau)$.
Compatibility with $\tau$ means that the induced representation on
$Q_\lambda^\tau$ is naturally isomorphic to this conjugate output:
applying $\tau$ before or after the construction gives the same result
up to this identification.

\begin{lemma}[Integral lattice and base change]\label{lem:lattice}
Suppose $\lambda$ is a root of unity. The output preserves a full
$\mathcal O_E$-lattice in $Q_\lambda$. The construction is compatible with every
embedding $\tau:E\hookrightarrow\C$ in the sense just defined.
\end{lemma}
\begin{proof}
For $G=\rho(F_n)$, set
$\Lambda=\sum_{g\in G}g\mathcal O_E^{\dim_E V}$, a full $G$-stable
lattice. Since $\lambda\in\mathcal O_E^\times$, the block formulas
for $A_i^{\pm1}$ show that they preserve $\Lambda^{\oplus n}$.
Thus the output preserves the full lattice
$\Lambda^{\oplus n}/(\Lambda^{\oplus n}\cap(K+L))\subset Q_\lambda$.
Extension of scalars along $\tau$ commutes with kernels, finite
intersections, sums, and quotients, hence with the output construction.
\end{proof}

\begin{theorem}[Componentwise finite-image criterion]\label{thm:free}
Let $E\subset\C$ be any number field over which the finite-image input
$\rho$ is defined, and assume that $\lambda\in E\setminus\{1\}$ is a
root of unity. For each field embedding $\tau:E\hookrightarrow\C$,
let $\mathcal I_\tau$ index the distinct isomorphism classes of
irreducible $\C[F_n]$-modules occurring in $V^\tau$.
For each $\alpha\in\mathcal I_\tau$, choose a representative
$U_\alpha$ and a positive definite $F_n$-invariant form $h_\alpha$ on it,
and put
\[
 W_\alpha:=\MCDR_{\tau(\lambda)}(U_\alpha),\qquad
 \mathcal H_\alpha:=\mathcal H_{\tau(\lambda),h_\alpha}.
\]
Here $\mathcal H_\alpha$ is the canonical quotient form of
Definition~\ref{def:quotient-form} on $W_\alpha$; dependence on $\tau$
is suppressed in the subscripts.
Then $\MCDR_\lambda(\rho)(F_n)$ is finite if and only if, for every $\tau$
and every $\alpha\in\mathcal I_\tau$ with $W_\alpha\ne0$, the form
$\mathcal H_\alpha$ is positive or negative definite.
This condition is independent of the choices of $h_\alpha$.
No rigidity, braid extension, or irreducibility of the full input or
output is assumed.
\end{theorem}
\begin{proof}
Since the input image is finite, averaging gives a positive definite
$F_n$-invariant form $h_\alpha$ on each $U_\alpha$.
By Schur's lemma, this form is unique up to a positive scalar.
Rescaling $h_\alpha$ rescales $\mathcal H_\alpha$ by the same positive
factor, so its definiteness does not depend on this choice.

Fix an embedding $\tau$ and define the multiplicity spaces
\[
 \mathcal M_\alpha:=\operatorname{Hom}_{\C[F_n]}(U_\alpha,V^\tau).
\]
The group $F_n$ acts trivially on $\mathcal M_\alpha$, whose dimension
is the multiplicity of $U_\alpha$. The input and output decompose as
\[
 V^\tau\simeq\bigoplus_{\alpha\in\mathcal I_\tau}
 U_\alpha\otimes_\C\mathcal M_\alpha,\qquad
 Q_\lambda^\tau\simeq\bigoplus_{\alpha\in\mathcal I_\tau}
 W_\alpha\otimes_\C\mathcal M_\alpha.
\]
The output decomposition follows from additivity of the construction,
by taking kernels and quotients componentwise.
By Lemma~\ref{lem:simple}, every nonzero $W_\alpha$ is irreducible.
If the output image is finite, so are its conjugates and their
constituents. Averaging produces a positive invariant form on each
$W_\alpha$. Schur's lemma implies that its nondegenerate canonical
Hermitian form is a nonzero real multiple of this positive form, hence
definite.

Conversely, change the sign of the canonical form on each negative
component and take its tensor product with a positive form on the
multiplicity space. Their direct sum is a positive invariant form for
the free-group output at every embedding. By Lemma~\ref{lem:lattice},
the image also preserves a full arithmetic lattice. We now apply the
classical arithmetic finiteness argument combining an invariant lattice
with positive definite invariant forms at every embedding; see
\cite[Sections 4.3 and 5]{RRV} for the same argument in the setting of
Goursat local systems. We recall the argument here. In the real vector
space obtained by restriction of scalars, these forms give a positive
invariant real quadratic form $q$. More explicitly, choose a
$\Z$-basis $e_1,\ldots,e_D$ of the resulting full $\Z$-lattice,
where $D=[E:\Q]\dim_E Q_\lambda$. For every element $g$ of the
output image, write
\[
 g e_j=\sum_{i=1}^D a_{ij}(g)e_i,\qquad a_{ij}(g)\in\Z.
\]
Invariance gives $q(g e_j)=q(e_j)$. To see why this bounds the
coordinates, write $q(v)=v^{\mathsf T}Hv$ in the real coordinates
defined by the lattice basis. If $D=0$, the output is zero and its
image is trivial, so assume $D>0$. The real symmetric matrix $H$ is
positive definite, hence its smallest eigenvalue $m$ is positive and
\[
 q(v)\ge m\|v\|^2,\qquad
 \|g e_j\|^2\le \frac{q(e_j)}{m}.
\]
Thus the level set $\{v:q(v)=q(e_j)\}$ is bounded; it is also closed
by continuity of $q$, and hence compact in this finite-dimensional
real space. For each $j$, the integer coordinates
$a_{ij}(g)$ are bounded uniformly in $g$, so only finitely many images
of each basis vector are possible. A linear map is determined by these
images, hence the output image is finite. Equivalently, a group
preserving a full $\Z$-lattice is finite if and only if its matrix
entries in a fixed lattice basis are uniformly bounded. This is the
concrete form of the compact--discrete intersection argument.
\end{proof}

The all-embedding condition is independent of the choice of coefficient
field. Indeed, for a finite extension $E'/E$, every embedding of $E$
into $\C$ extends to $E'$. The conjugate input and output for an
embedding of $E'$ depend, up to scalar extension, only on its restriction
to $E$. Their complex irreducible constituents and definiteness tests
are therefore the same. Enlarging $E$ repeats these tests without adding
new conditions. In particular, coefficient fields for the same matrix
model can be compared inside their compositum; neither field is required
to contain all eigenvalues of the input matrices or their products.

\begin{corollary}[Eigenangle formulation]\label{cor:angles}
Under the assumptions of Theorem~\ref{thm:free}, fix an embedding
$\tau:E\hookrightarrow\C$ and an index $\alpha\in \mathcal I_\tau$.
Let $\rho_\alpha:F_n\to\GL(U_\alpha)$ denote the corresponding
irreducible representation and put $N:=\dim_\C U_\alpha$.
Choose a positive definite invariant Hermitian form on $U_\alpha$
(which exists because its image is finite) and an orthonormal basis.
In this corollary, write
$g_j:=\rho_\alpha(x_j)\in U(N)$ for $1\le j\le n$, and put
$P_n:=g_1\cdots g_n$.
These are the matrices of the chosen constituent, not the matrices
on the full input space.

Following the eigenangle notation of \cite{Signature}, let
$a_j^{(k)}\in[0,1)$, $1\le j\le n$, $1\le k\le N$, be defined by
\[
 \Spec(g_j)
 =\{\exp(2\pi i a_j^{(k)}):1\le k\le N\},
\]
where spectra are listed with algebraic multiplicity. Let $l\in(0,1)$
be the unique angle with $\tau(\lambda)=\exp(2\pi i l)$, and let
$\beta_k(l)\in[0,1)$, $1\le k\le N$, be the eigenangles of
$\tau(\lambda)P_n$, again with algebraic multiplicity. The index $k$
only enumerates each spectrum; it does not identify eigenvectors of
different matrices. Set
\[
 \kappa:=\sum_{j=1}^n\dim_\C\ker(g_j-I_N),\qquad R:=nN-\kappa,
 \qquad r(l):=\dim_\C\ker(\tau(\lambda)P_n-I_N),
\]
\[
 m(l):=Nl+\sum_{j=1}^n\sum_{k=1}^N a_j^{(k)}
              -\sum_{k=1}^N\beta_k(l)\in\Z.
\]
All these quantities refer to the fixed pair $(\tau,\alpha)$.
The criterion in Theorem~\ref{thm:free} is precisely
\[
 m(l)\in\{r(l),R\}
 \quad\text{for every }\tau:E\hookrightarrow\C
 \text{ and every }\alpha\in \mathcal I_\tau\text{ with }R-r(l)>0.
\]
\end{corollary}
\begin{proof}
Applying \cite[Theorem~1.1]{Signature} to the chosen constituent
and the parameter $\tau(\lambda)$ gives canonical signature
$(R-m(l),m(l)-r(l))$ and output dimension $R-r(l)$.
Definiteness is exactly the vanishing of one of these two inertia
indices; zero-dimensional outputs impose no condition.
The determinant identity
\[
 \exp\Bigl(2\pi i\sum_{k=1}^N\beta_k(l)\Bigr)
 =\tau(\lambda)^N\det P_n
 =\exp\Bigl(2\pi i\bigl(Nl+\sum_{j=1}^n\sum_{k=1}^N a_j^{(k)}\bigr)\Bigr)
\]
makes $m(l)$ an integer. Since the input
image is finite and $\lambda$ is a root of unity, all the eigenangles
above are rational. Thus this is a finite exact check for a fixed
input and a fixed parameter.
\end{proof}

For later use, let $\gamma^{(k)}\in[0,1)$ ($1\le k\le N$) be the
eigenangles of $P_n$, counted with multiplicity. After reindexing,
$\beta_k(l)=\{l+\gamma^{(k)}\}$, and therefore
\begin{equation}\label{eq:floor-signature}
\begin{aligned}
 m(l)&=c+\sum_{k=1}^N\lfloor l+\gamma^{(k)}\rfloor,\\
 c&=\sum_{j=1}^n\sum_{k=1}^N a_j^{(k)}-\sum_{k=1}^N\gamma^{(k)}\in\Z,\\
 r(l)&=\#\{k:1\le k\le N,\ l+\gamma^{(k)}=1\}.
\end{aligned}
\end{equation}
The last equality uses the diagonalizability of the unitary matrix $P_n$;
integrality of $c$ follows from $\det P_n=\prod_j\det g_j$.
These are the floor-function formulas used in
Corollary~\ref{cor:unitary-locus} and Lemma~\ref{lem:gap}.

\begin{remark}[What spectral data are required]
The angles in Corollary~\ref{cor:angles} belong to each irreducible input
constituent, not just to the total reducible tuple. Identifying these
constituents uses the common representation of the finite input group.
There is no assumption of simultaneous diagonalizability of the matrices
$g_j$. The spectrum of their ordered product is separate required data.
The sign can vary with the constituent and with the coefficient embedding.
\end{remark}

\section{Braid group representations: lifting unitarizability and finiteness}\label{sec:full}

\begin{definition}[Katz--Long--Moody representation]\label{def:klm}
Let $\rho:\Gamma=F_n\rtimes B_n\to\GL(V)$ be a finite-dimensional
complex representation, with the Artin action \eqref{eq:artin-action},
and let $\lambda\in\C\setminus\{0,1\}$. Put $N:=\dim_\C V$ and write
$s_i:=\rho(\sigma_i)$ for $1\le i\le n-1$ and
$g_j:=\rho(x_j)$ for $1\le j\le n$. Let $I_d$ denote the $d\times d$
identity matrix and $0_N$ the $N\times N$ zero matrix.
For $1\le i\le n-1$, the intermediate braid operator on $V^{\oplus n}$ is
\begin{equation}\label{eq:Si}
 S_i=s_i^{\oplus n}\operatorname{diag}
       (I_{(i-1)N},\Theta_i,I_{(n-i-1)N}),\qquad
 \Theta_i=\begin{pmatrix}0_N&g_i\\I_N&I_N-g_{i+1}\end{pmatrix}.
\end{equation}
Here $\Theta_i$ is a $2N\times2N$ block acting on the $i$th and $(i+1)$st
copies of $V$; it is the block denoted by $R_i$ in
\cite[Section 2.2, Definition 8]{KLM}, renamed here to distinguish it from the
spectral integers $R$ and $R_\alpha$. The separate notation
$\Theta_i(\cdot)$ of that reference is not used here. A block $I_0$ is omitted when $i=1$ or $i=n-1$.
Together with the $A_i$ of Definition~\ref{def:dr-convolution}, these
operators induce the \emph{Katz--Long--Moody representation}
$T_\lambda=\KLM_\lambda(\rho):\Gamma\to\GL(Q_\lambda)$,
defined by $T_\lambda(x_j)[v]=[A_jv]$ ($1\le j\le n$) and
$T_\lambda(\sigma_i)[v]=[S_iv]$ ($1\le i<n$) \cite{KLM}.
\end{definition}

\begin{lemma}[Full KLM invariance]\label{lem:full-invariance}
Suppose that the full input $\rho:F_n\rtimes B_n\to\GL(V)$ preserves
a positive definite Hermitian form $h$, and let $|\lambda|=1$,
$\lambda\ne1$. Then the canonical quotient form of
Definition~\ref{def:quotient-form} is invariant under the full KLM representation
$T_\lambda$.
\end{lemma}
\begin{proof}
The intermediate form is invariant under both the free generators $A_j$ and
the braid operators $S_i$ by \cite[Theorem 8]{KLM}. Its radical is
$K+L$ by Lemma~\ref{lem:forms}(i), so the form descends to the quotient;
the induced form is nondegenerate and invariant under $T_\lambda$.
This is also the quotient-form statement of \cite[Theorem 9]{KLM}.
\end{proof}

\begin{theorem}[Lifting positive unitarizability]\label{thm:unitary-lift}
Suppose the full input $\rho:\Gamma\to\GL(V)$ preserves a positive
definite Hermitian form $h$. Let $|\lambda|=1$, $\lambda\ne1$.
Let $\mathcal I_F$ index the distinct isomorphism classes of irreducible
constituents of the restriction $\rho|_{F_n}$. For each
$\alpha\in\mathcal I_F$, choose a representative $F_n$-module
$U_\alpha$ and a positive definite $F_n$-invariant form $h_\alpha$ on it.
Put $W_\alpha:=\MCDR_\lambda(U_\alpha)$, and let $\mathcal H_\alpha$
be its canonical quotient form constructed from $h_\alpha$, with
the same choice of $\lambda^{1/2}$ for every $\alpha\in\mathcal I_F$.
The following conditions are equivalent:
\begin{enumerate}[label=(\roman*)]
\item $T_\lambda|_{F_n}$ is unitarizable;
\item for every $\alpha\in \mathcal I_F$ with $W_\alpha\ne0$, the form
$\mathcal H_\alpha$ is definite;
\item $T_\lambda$ is unitarizable as a $\Gamma$-representation.
\end{enumerate}
The input and output may be reducible and have infinite image.
\end{theorem}
\begin{proof}
Since $h$ is positive and the input is unitary for $h$, the orthogonal
complement of every $F_n$-invariant subspace is invariant. Finite
dimensionality therefore gives complete reducibility. Define
\[
 \mathcal M_\alpha:=\operatorname{Hom}_{\C[F_n]}(U_\alpha,V),
 \qquad d_\alpha:=\dim_\C\mathcal M_\alpha,
\]
and let $V_\alpha\subseteq V$ be the $U_\alpha$-isotypic component,
the sum of all $F_n$-submodules isomorphic to $U_\alpha$.
The group $F_n$ acts trivially on $\mathcal M_\alpha$, and $d_\alpha$
is the multiplicity of $U_\alpha$. Evaluation gives the orthogonal
isotypic decomposition
\[
 V=\bigoplus_{\alpha\in\mathcal I_F}V_\alpha,\qquad
 V_\alpha\simeq U_\alpha\otimes_\C\mathcal M_\alpha.
\]
Here $U_\alpha$ is a representation space for the free-group
restriction, not the image group $\rho(F_n)$ or $\rho(\Gamma)$.
By Schur's lemma, the restriction of $h$ has the form
$h|_{V_\alpha}=h_\alpha\otimes k_\alpha$ for a positive definite
form $k_\alpha$ on $\mathcal M_\alpha$.
None of this requires finite image.
Write $\mathcal H:=\mathcal H_{\lambda,h}$ for the canonical form
on the full output, using the same choice of $\lambda^{1/2}$.

The spaces $W_\alpha$ are the convolved constituents.
Let $Q_\alpha$ be the corresponding output block, the image of $V_\alpha^{\oplus n}$ in
$Q_\lambda=V^{\oplus n}/(K+L)$. The construction and
Lemma~\ref{lem:forms} give
\[
 Q_\lambda=\bigoplus_{\alpha\in \mathcal I_F} Q_\alpha,
 \qquad Q_\alpha\simeq W_\alpha\otimes \mathcal M_\alpha,
 \qquad \mathcal H|_{Q_\alpha}=\mathcal H_\alpha\otimes k_\alpha.
\]
By Lemma~\ref{lem:simple}, every nonzero $W_\alpha$ is irreducible.
If (i) holds, restricting a positive invariant form to a copy of
$W_\alpha$ gives a positive invariant form there. Its nondegenerate
canonical form is a nonzero real multiple by Schur's lemma. Hence
$\mathcal H_\alpha$ is definite, proving (ii).

Assume (ii). The possibly infinite group $G_F=\rho(F_n)$ is normal
in $\rho(\Gamma)$. Thus $\rho(\Gamma)$ permutes the $G_F$-isotypic
spaces $V_\alpha$. For $\eta\in\Gamma$, denote by
$\eta\cdot\alpha\in \mathcal I_F$ the index determined by
$\rho(\eta)V_\alpha=V_{\eta\cdot\alpha}$. The $g_j$ preserve each such space, and
\eqref{eq:Si} sends $V_\alpha^{\oplus n}$ to the direct sum of copies
of the permuted input space. Since $K$ and $L$ decompose over these
spaces, the full output permutes the $Q_\alpha$.

Each nonzero restriction $\mathcal H|_{Q_\alpha}$ is definite, including
when $\dim \mathcal M_\alpha>1$. Let $\epsilon_\alpha\in\{1,-1\}$ be its
sign. By Lemma~\ref{lem:full-invariance}, every output group element is an
isometry between the blocks it permutes, so
$\epsilon_{\eta\cdot\alpha}=\epsilon_\alpha$. It follows that
\[
 \mathcal H^+=\bigoplus_{\substack{\alpha\in \mathcal I_F\\Q_\alpha\ne0}}
     \epsilon_\alpha\mathcal H|_{Q_\alpha}
\]
is positive definite and invariant under $\Gamma$, proving (iii).
The implication (iii)$\Rightarrow$(i) follows by restriction.
The zero output satisfies all three conditions with the usual convention.
\end{proof}

\begin{remark}[How braids permute the components]\label{rem:braid-components}
In the proof of Theorem~\ref{thm:unitary-lift}, normality of $F_n$
is already part of the semidirect-product setting. Its role is to
control where a braid sends an input component. Indeed, for
$b\in B_n$, $f\in F_n$, and $u$ in an $F_n$-invariant subspace $U\subseteq V$,
\[
 \rho(f)\rho(b)u=\rho(b)\rho(b^{-1}fb)u\in\rho(b)U,
\]
since $b^{-1}fb\in F_n$. Thus $\rho(b)U$ is again $F_n$-invariant;
if $U$ is irreducible, so is $\rho(b)U$. Conjugation by $b$ therefore
permutes the isomorphism classes in $\mathcal I_F$ and sends each entire
isotypic component to another one. With the notation of the proof,
\[
 \rho(b)V_\alpha=V_{b\cdot\alpha},\qquad
 T_\lambda(b)Q_\alpha=Q_{b\cdot\alpha}
 \quad(b\in B_n,\ \alpha\in \mathcal I_F).
\]
The second equality uses the KLM block operators \eqref{eq:Si} and
the componentwise decomposition of $K+L$; it is a statement about
the corresponding output blocks, not just the input spaces.

Under condition (ii), invariance of the canonical form means that
positive definite blocks cannot be sent to negative definite blocks.
Thus changing the sign on all negative blocks preserves braid
invariance, as well as the $F_n$-invariance of each block. This is the
reason that the form $\mathcal H^+$ constructed in the proof is
positive definite and invariant under the full semidirect product.
\end{remark}

For a fixed positively unitarizable input, the following corollary
uses the signature formula to determine exactly which unit-circle
parameters $\lambda\ne1$ give a positively unitarizable full KLM output, by testing
finitely many intervals and their intervening boundary points.

\begin{corollary}[Spectral walls for unitarizability]\label{cor:unitary-locus}
Let $\rho:\Gamma\to\GL(V)$ be unitarizable, and use the index set
$\mathcal I_F$ and the irreducible $F_n$-modules $U_\alpha$ of
Theorem~\ref{thm:unitary-lift}. For each $\alpha\in \mathcal I_F$, put
$N_\alpha:=\dim_\C U_\alpha$ and let
$\rho_\alpha:F_n\to\GL(U_\alpha)$ be the corresponding representation.
Choose an orthonormal basis for a positive definite invariant form on
$U_\alpha$, and write
\[
 g_{j,\alpha}:=\rho_\alpha(x_j)\in U(N_\alpha)
 \quad(1\le j\le n),\qquad
 P_{n,\alpha}:=g_{1,\alpha}\cdots g_{n,\alpha}.
\]
Following \cite{Signature}, let $a_{j,\alpha}^{(k)}\in[0,1)$ be the
eigenangles of $g_{j,\alpha}$ and $\gamma_\alpha^{(k)}\in[0,1)$ those
of $P_{n,\alpha}$, each listed with algebraic multiplicity, where
\[
 \alpha\in \mathcal I_F,\qquad 1\le j\le n,\qquad 1\le k\le N_\alpha.
\]
Here $k$ enumerates each spectrum separately; no common eigenbasis is
assumed. Set
\[
 R_\alpha:=nN_\alpha-\sum_{j=1}^n\dim_\C\ker(g_{j,\alpha}-I_{N_\alpha}),
 \qquad
 c_\alpha:=\sum_{j=1}^n\sum_{k=1}^{N_\alpha}a_{j,\alpha}^{(k)}
              -\sum_{k=1}^{N_\alpha}\gamma_\alpha^{(k)}\in\Z.
\]
For $\lambda=e^{2\pi il}$ with $0<l<1$, define
\[
 m_\alpha(l):=c_\alpha+
       \sum_{k=1}^{N_\alpha}\lfloor l+\gamma_\alpha^{(k)}\rfloor,
 \qquad
 r_\alpha(l):=\#\{k\in\{1,\ldots,N_\alpha\}:
                      l+\gamma_\alpha^{(k)}=1\}.
\]
Then $T_{e^{2\pi il}}$ is unitarizable as a $\Gamma$-representation
if and only if
\[
 m_\alpha(l)\in\{r_\alpha(l),R_\alpha\}
 \quad\text{for every }\alpha\in \mathcal I_F
 \text{ with }R_\alpha-r_\alpha(l)>0.
\]
Call the points of the finite set
\[
 \{1-\gamma_\alpha^{(k)}:
    \alpha\in \mathcal I_F,\ 1\le k\le N_\alpha,\ \gamma_\alpha^{(k)}>0\}
\]
the \emph{spectral walls}, and the connected components of their
complement in $(0,1)$ the \emph{open intervals}. A wall is precisely an
angle $l$ for which $e^{2\pi il}$ is exceptional for $\rho|_{F_n}$ in
the sense defined in Section~\ref{sec:resonance}.
The unitarizability locus is a union of some open intervals and some
spectral walls. No arithmetic or root-of-unity assumption is required.
\end{corollary}
\begin{proof}
The signature formula \cite[Theorem~1.1]{Signature}, with eigenangles
$\{l+\gamma_\alpha^{(k)}\}$ of $\lambda P_{n,\alpha}$ (braces denote
fractional parts), gives
$\operatorname{sig}\mathcal H_\alpha=
(R_\alpha-m_\alpha(l),m_\alpha(l)-r_\alpha(l))$.
The criterion follows from Theorem~\ref{thm:unitary-lift}.
The floor functions are constant between the listed walls, proving
the description of the locus; integrality of $c_\alpha$ follows by
taking determinants of $P_{n,\alpha}$.
\end{proof}

\begin{corollary}[Arithmetic consequence: finiteness of the full image]
\label{thm:full}
Let $V$ be a finite-dimensional complex vector space and let
$\rho:\Gamma\to\GL(V)$ have finite image.
For every $\lambda\in\C\setminus\{0,1\}$,
\[
 T_\lambda(F_n)\text{ is finite}
 \quad\Longleftrightarrow\quad
 T_\lambda(\Gamma)\text{ is finite}.
\]
There is no irreducibility assumption.
\end{corollary}
\begin{proof}
The reverse implication is immediate. If the input is trivial on $F_n$,
the output is zero. Otherwise Proposition~\ref{prop:nonroot} shows that
finiteness of $T_\lambda(F_n)$ forces $\lambda$ to be a root of unity.
Since $\rho$ factors through a finite group, choose a basis in which
its matrices are defined over a number field $E$, and replace $E$ by
$E(\lambda)$. For the arithmetic argument below, use the resulting
$E$-form of $V$.

At every embedding $\tau:E\hookrightarrow\C$, the conjugate input
has finite full image and the conjugate output has finite free image.
Average over the finite conjugate input to obtain a positive definite invariant
input form. The finite free output is unitarizable, so
Theorem~\ref{thm:unitary-lift} supplies a positive form for the full
conjugate output. Construct the input lattice by averaging over
$\rho(\Gamma)$ in Lemma~\ref{lem:lattice}. The operators $S_i$ and
$S_i^{-1}$ also preserve its $n$-fold sum: their blocks use only
$s_i^{\pm1}$, $g_j^{\pm1}$ and integers. Hence the full output has a
full invariant arithmetic lattice. The restriction-of-scalars argument
from Theorem~\ref{thm:free} proves finiteness.
\end{proof}

\begin{remark}[Pure braid scope]
The natural surjective homomorphism $\Gamma\to B_n\to \mathfrak S_n$ sends $(f,b)$ to
the permutation induced by $b$ and has kernel $F_n\rtimes PB_n\simeq PB_{n+1}$.
Thus $\Gamma/(F_n\rtimes PB_n)\simeq \mathfrak S_n$, and this subgroup has
index $n!$ in $\Gamma$. For any representation of $\Gamma$, its image
is finite if and only if the image of this subgroup is finite.
Under the finite-input assumption of Corollary~\ref{thm:full}, these
conditions are also equivalent to finiteness of $T_\lambda(F_n)$. This does
not assume an extension to $B_{n+1}$. It does, however, start with an
input defined on $F_n\rtimes B_n$.

For an input defined only on $PB_{n+1}$, the proof must instead verify
the analogous pure-group invariance of the form, the permutation of
input isotypic spaces, and the integral lattice for the transformed pure braid matrices
in Haraoka's multiplicative middle convolution
\cite{KLM,Haraoka2020}.
If these hold, the identical sign argument applies, with $PB_{n+1}$
in place of $\Gamma$. We do not replace this verification by evaluating
the given representation on half-twists outside its domain.
The free-group theorem of Section~\ref{sec:free} is independent of all these issues.
\end{remark}

\begin{remark}[Braid restriction versus the whole group]
Corollary~\ref{thm:full} does not identify finiteness of $T_\lambda(B_n)$
alone with finiteness of the free image. For example, take $n=2$,
$g_1=g_2=i$, $s_1=1$, and $\lambda=-i$. Here $K=L=0$ and
\[
 A_1=\begin{pmatrix}1&i-1\\0&1\end{pmatrix},\qquad
 S_1=\begin{pmatrix}0&i\\1&1-i\end{pmatrix}.
\]
$A_1$ has infinite order, whereas $S_1$ is diagonalizable with
eigenvalues $1,-i$ and has order four. The full input is finite.
Thus the output can have infinite image on $F_n$ but finite image on
$B_n$, even when the full input has finite image. This distinction
motivates the separate study of the braid restriction in
Section~\ref{sec:resonance}, where we analyze exceptional parameters
and prove that finiteness of the braid output is independent of
$\lambda\in\C\setminus\{0,1\}$ whenever the input braid image is finite.
\end{remark}

\begin{remark}[The excluded parameter $1$]
If one defines the same coordinate quotient at $\lambda=1$, use
$D:V^{\oplus n}\to V$ from \eqref{eq:D-definition}. Then
$A_i-I_{nN}=\iota_iD$, where $\iota_i:V\to V^{\oplus n}$ is the
inclusion into the $i$th coordinate ($1\le i\le n$). We have
$L=\ker D$ and $K\subseteq L$. Moreover
$DA_i=g_iD$ and $DS_i=s_iD$. Thus this coordinate quotient is the
input subrepresentation $\operatorname{im}D\subseteq V$ and has
finite image for finite full input. This is separate from the
signature theorem, whose stated domain excludes $\lambda=1$.
\end{remark}

\begin{corollary}[Spectral criterion for the full output]\label{cor:full-spectral}
Let $E$ be a number field and let
$\rho:\Gamma\to\GL_E(V)$ have finite full image and nontrivial
restriction to $F_n$. For $\lambda\ne0,1$, the following are equivalent:
\begin{enumerate}[label=(\roman*)]
\item $T_\lambda(\Gamma)$ is finite;
\item $T_\lambda(F_n\rtimes PB_n)$ is finite;
\item $\lambda$ is a root of unity and the componentwise eigenangle
condition in Corollary~\ref{cor:angles} holds at every embedding of a
coefficient field containing $\lambda$.
\end{enumerate}
\end{corollary}
\begin{proof}
Put $\Gamma_{\rm pure}:=F_n\rtimes PB_n$. The quotient $\Gamma/\Gamma_{\rm pure}\simeq \mathfrak S_n$
surjects onto $T_\lambda(\Gamma)/T_\lambda(\Gamma_{\rm pure})$, so
\[
 [T_\lambda(\Gamma):T_\lambda(\Gamma_{\rm pure})]\le n!.
\]
Consequently, if (ii) holds, then
$|T_\lambda(\Gamma)|\le n!\,|T_\lambda(\Gamma_{\rm pure})|<\infty$, proving (i).
The converse follows by restriction. This equivalence does not require
the finite-input assumption.
Under the stated assumptions, the equivalence with (iii) follows from
Theorem~\ref{thm:free}, Corollary~\ref{thm:full}, and
Proposition~\ref{prop:nonroot}.
\end{proof}

\section{Finite-image parameters: classification and effective bounds}\label{sec:parameters}
Sections~\ref{sec:free} and~\ref{sec:full} give finite-image criteria
for a specified parameter $\lambda$. We now classify the complete
set of such parameters for a fixed input and give an effective order
bound whenever that set is finite. For a nontrivial finite-image input,
Proposition~\ref{prop:nonroot} restricts the candidates to roots of
unity, but still leaves infinitely many possible orders.

Throughout this section, let $\rho:F_n\to\GL_E(V)$ have finite image,
where $E$ is a number field and $n\ge2$.
An input matrix $g_j=\rho(x_j)$ ($1\le j\le n$) is \emph{active}
on an irreducible constituent $U\subseteq V$ if $g_j|_U\ne I_U$;
thus we count the matrices that act nontrivially on $U$.
If some irreducible constituent has at least two active generators,
Proposition~\ref{prop:order-bound}\textup{(ii)} gives an effective upper bound on
the orders of all finite-output parameters. The criterion of
Theorem~\ref{thm:free} then yields a complete finite enumeration.
Theorem~\ref{thm:parameters} combines this conclusion with the
classification of the remaining inputs, for which the parameter set
is infinite.

The proof first finds an open interval of eigenangles on which the
canonical form of such a constituent is indefinite. A root of unity
of sufficiently large order has a conjugate in this interval under
an automorphism preserving the input constituent's isomorphism class.
Finite image at the original parameter would imply finite image at
this conjugate, contradicting the definiteness criterion.
An effective count of these conjugates bounds the possible orders.
The arithmetic test at a fixed rational parameter has precedents in
\cite{Haraoka1994,RRV,GV,BH,BeukersA}; the assertion developed here is a
bound covering all finite-output parameters for a fixed input, together
with the classification of the infinite parameter sets. The argument uses the standard fact that a complex representation of a
finite group is determined up to isomorphism by its character
\cite[Chapter~2]{SerreRep}: the constituent's character takes values in
$\Q(\zeta_M)$, where $M$ is the exponent of the finite input group,
so automorphisms fixing this field preserve its isomorphism class.
This requires neither rigidity nor simultaneous diagonalization.

\subsection{An interval of indefinite forms}
\begin{lemma}\label{lem:nonzero}
If $U$ is an irreducible complex input of dimension $N$ with at least
two active generators, then $\MCDR_\lambda(U)\ne0$ for every $\lambda\ne0,1$.
\end{lemma}
\begin{proof}
If $K+L_\lambda=U^{\oplus n}$, the images of the common fixed vectors
$L_\lambda$ span $U^{\oplus n}/K$. Thus each $A_i$ induces the identity
on this quotient. Applying $A_i-I_{nN}$ to a vector supported in coordinate
$j\ne i$ shows
\[
 (g_i-I_N)(g_j-I_N)=0 \qquad(i\ne j),
\]
since the off-diagonal scalar is either $1$ or $\lambda$, both nonzero.
The same relation with $i,j$ reversed implies that all $g_i$ commute.
An irreducible finite-dimensional complex representation of this
commuting family is one-dimensional. In dimension one, the displayed
identity precludes two active generators, a contradiction.
\end{proof}

For the next lemma, let $U$ be a unitarizable $F_n$-module with
fixed positive definite invariant form $h$, and put $N:=\dim_\C U$.
In an orthonormal basis, write $g_j$ for the action of $x_j$
($1\le j\le n$). Use $\mathcal H(l)$, $B$, and $D:U^{\oplus n}\to U$
from Definition~\ref{def:ambient-pencil}, applied to $U$ in place of $V$,
and write $\mathcal H_{\lambda,h}$ for the canonical quotient form of
Definition~\ref{def:quotient-form}.
Because $\mathcal H(l)$ is defined on the fixed space $U^{\oplus n}$
before taking the parameter-dependent quotient, the continuous
coefficients in \eqref{eq:canonical-pencil} extend it to $l=0,1$,
where we will compare the endpoint forms.

\begin{lemma}\label{lem:gap}
Let $U$ be a unitarizable irreducible input with at least two active
generators. For $\lambda=e^{2\pi il}$ with $0<l<1$, use the
canonical quotient form $\mathcal H_{\lambda,h}$ specified above.
There is a nonempty open interval $I\subset(0,1)$ on which $\mathcal H_{\lambda,h}$
is nondegenerate and indefinite.
If the input representation has finite image, an interval $I$ with
rational endpoints is obtained by the following finite procedure.
\begin{enumerate}[label=\textup{Step \arabic*.},leftmargin=*,font=\normalfont]
\item Since the input image is finite, all eigenangles below are rational.
Compute the eigenangles $a_j^{(\nu)}\in[0,1)$ of $g_j$ and
$\gamma^{(\nu)}\in[0,1)$ of $P_n=g_1\cdots g_n$
($1\le j\le n$, $1\le\nu\le N$), all counted with multiplicity. The index $\nu$
plays the role of $k$ in Corollary~\ref{cor:angles} and enumerates
each spectrum separately.
\item Compute the integers
\[
 R=nN-\sum_{j=1}^n\#\{\nu:a_j^{(\nu)}=0\},
 \qquad
 c=\sum_{j=1}^n\sum_{\nu=1}^N a_j^{(\nu)}-\sum_{\nu=1}^N \gamma^{(\nu)}.
\]
\item Remove repetitions from
$\{0,1\}\cup\{1-\gamma^{(\nu)}:1\le\nu\le N,\ \gamma^{(\nu)}\ne0\}$
and sort the resulting points as $0=t_0<t_1<\cdots<t_s=1$.
The interior points $t_1,\ldots,t_{s-1}$ are precisely all
$l\in(0,1)$ for which $L_{e^{2\pi il}}\ne0$; hence
$L_{e^{2\pi il}}=0$ on each interval $(t_{k-1},t_k)$ ($1\le k\le s$).
\item For $k=1,\ldots,s$, set $l_k=(t_{k-1}+t_k)/2$ and compute
\[
 m_k=c+\sum_{\nu=1}^N\lfloor l_k+\gamma^{(\nu)}\rfloor.
\]
By the signature formula, written as in \eqref{eq:floor-signature},
the signature throughout
$(t_{k-1},t_k)$ is $(R-m_k,m_k)$.
\item Select the first $k$ for which $0<m_k<R$ and return
$I=(t_{k-1},t_k)$. Such a $k$ exists by the assertion above.
\end{enumerate}
\end{lemma}
\begin{proof}
We justify the five steps in order.

\smallskip\noindent\textit{Step 1: eigenangles.}
Each $g_j$ and their product $P_n$ are unitary, so their eigenangles
lie in $[0,1)$. If the input image is finite, these matrices have finite
order, so all the eigenangles are rational.

\smallskip\noindent\textit{Step 2: the constants $R$ and $c$.}
Since $\dim\ker(g_j-I_N)$ is the multiplicity of the eigenangle zero,
the stated value of $R$ is $\dim W$, where $W:=U^{\oplus n}/K$
is the intermediate quotient for the input $U$.
Moreover, $\det P_n=\prod_{j=1}^n\det g_j$ gives
$\exp(2\pi i c)=1$, hence $c\in\mathbb Z$.

\smallskip\noindent\textit{Step 3: the open intervals.}
For $\lambda=e^{2\pi il}$ with $0<l<1$, the space $L_\lambda$
is naturally isomorphic to $\ker(\lambda P_n-I_N)$.
Thus $L_\lambda\ne0$ exactly when $l=1-\gamma^{(\nu)}$ for
some nonzero eigenangle $\gamma^{(\nu)}$. Removing repetitions
therefore lists every such interior point exactly once as
$t_1,\ldots,t_{s-1}$. On each intervening interval
$(t_{k-1},t_k)$, we have $L_\lambda=0$.
Since $\operatorname{rad}\mathcal H(l)=K+L_\lambda$, the form induced
on $W$ is nondegenerate there and coincides with $\mathcal H_{\lambda,h}$.
There are finitely many such intervals.

\smallskip\noindent\textit{Step 4: the signature on each interval.}
For finite input, all $\gamma^{(\nu)}$ are rational, so the endpoints
$t_k$ and the midpoints $l_k=(t_{k-1}+t_k)/2$ are rational.
The signature formula and its equivalent floor-function expression
\eqref{eq:floor-signature} give
\[
 \operatorname{sig}\mathcal H_{\lambda,h}=(R-m(l),m(l)-r(l)),\qquad
 m(l)=c+\sum_{\nu=1}^N\lfloor l+\gamma^{(\nu)}\rfloor,\qquad
 r(l)=\#\{\nu:l+\gamma^{(\nu)}=1\}.
\]
On $(t_{k-1},t_k)$, we have $r(l)=0$, and each floor term is
constant. Consequently, evaluation at $l_k$ gives the signature
$(R-m_k,m_k)$ throughout the interval. This form is indefinite
exactly when $0<m_k<R$.

\smallskip\noindent\textit{Step 5: existence of an open interval of indefinite forms.}
Let $q_K:U^{\oplus n}\to W=U^{\oplus n}/K$ be the quotient map,
and write $[v]_K=q_K(v)$. The subspace $K$ is independent of $l$.
By Lemma~\ref{lem:forms}(i), it lies in the radical of $\mathcal H(l)$
for $0<l<1$; this remains true at $l=0,1$ by continuity of
\eqref{eq:canonical-pencil}. Thus, for every $l\in[0,1]$, define
\[
 \mathcal H^K(l)([v]_K,[w]_K):=v^*\mathcal H(l)w.
\]
This is independent of the representatives. For $0<l<1$ with
$L_\lambda=0$, where $\lambda=e^{2\pi il}$, it agrees with the
canonical quotient form in \eqref{eq:canonical-quotient-form}
under the natural identification $W=Q_\lambda$.

To verify continuity also at the spectral walls and endpoints,
fix a linear section $J:W\to U^{\oplus n}$ of $q_K$, so that
$q_KJ=I_W$. In fixed bases, the Hermitian matrix representing
$\mathcal H^K(l)$ is
\[
 J^*\mathcal H(l)J
 =\sin(\pi l)J^*BJ+\cos(\pi l)J^*D^*DJ.
\]
Every entry is continuous on $[0,1]$. Hence $\mathcal H^K(l)$ is
a continuous family on the fixed space $W$, even where it is degenerate.
At the endpoints,
\[
 \mathcal H^K(0)([v]_K,[w]_K)=(Dv)^*Dw,\qquad
 \mathcal H^K(1)([v]_K,[w]_K)=-(Dv)^*Dw.
\]
Since $D(K)=0$ and $D\ne0$, these are nonzero positive and negative
semidefinite forms, respectively: the first takes the value
$\|Dv\|^2\ge0$ on $([v]_K,[v]_K)$, and this value is positive
for some $v$.

Suppose, for a contradiction, that $\mathcal H^K(l)$ is not
indefinite for any $l\in\bigcup_{k=1}^{s}(t_{k-1},t_k)$.
By Step~3, the form would then be definite on each interval.
Continuity at $0$ forces the first interval to be positive definite,
and continuity at $1$ forces the last to be negative definite.
If there is only one interval, this is already a contradiction.
Otherwise, two adjacent intervals must have opposite signs.
At their common endpoint, continuity would make the limiting form
both positive and negative semidefinite, hence zero on $W$.
Its quotient by its radical would therefore be zero. This quotient
is $U^{\oplus n}/(K+L_\lambda)=\MCDR_\lambda(U)$,
contradicting Lemma~\ref{lem:nonzero}.
Thus there exists an open interval $(t_{k-1},t_k)$ on which
$\mathcal H^K(l)$ is indefinite for every $l$, equivalently $0<m_k<R$.
The finite search in Step~5 returns such an interval.

\end{proof}

\subsection{Middle convolution at Galois-conjugate parameters}
For a fixed finite-image input, the finite-image property of middle
convolution is unchanged when an algebraic convolution parameter is
replaced by a Galois conjugate over the input cyclotomic field.
More precisely, put $G=\rho(F_n)$ and let $M$ be its exponent.
Choose a basis in which $\rho$ is defined over a number field,
as explained in Section~\ref{sec:free-arithmetic}. For
$\lambda\in\overline{\Q}\setminus\{0,1\}$ and every field automorphism
$\sigma:\overline{\Q}\to\overline{\Q}$ fixing $\Q(\zeta_M)$ pointwise,
we have
\begin{equation}\label{eq:galois-finiteness}
 \MCDR_\lambda(\rho)(F_n)\text{ is finite}
 \quad\Longleftrightarrow\quad
 \MCDR_{\sigma(\lambda)}(\rho)(F_n)\text{ is finite}.
\end{equation}
Indeed, for each $x\in F_n$, let $\rho^\sigma(x)$ be obtained by
applying $\sigma$ to each entry of the representation matrix $\rho(x)$.
Every matrix in the finite input group $G$ has eigenvalues that are
$M$-th roots of unity. Hence $\sigma$ preserves the trace of every
input group element, so $\rho^\sigma\simeq\rho$ by the character theory
of finite groups \cite[Chapter~2]{SerreRep}. The matrix formulas and
the kernel and quotient operations defining middle convolution are
compatible with applying $\sigma$ entrywise to matrices and sending
$\lambda$ to $\sigma(\lambda)$. Consequently,
\begin{equation}\label{eq:galois-mc}
 \bigl(\MCDR_\lambda(\rho)\bigr)^\sigma
 \simeq \MCDR_{\sigma(\lambda)}(\rho^\sigma)
 \simeq \MCDR_{\sigma(\lambda)}(\rho).
\end{equation}
Applying a field automorphism entrywise preserves finiteness of a matrix
group, giving \eqref{eq:galois-finiteness}; the reverse direction also
follows by applying $\sigma^{-1}$. The same argument applies to every
irreducible input constituent. Here \emph{relative conjugates} refers
to the parameters $\sigma(\lambda)$ obtained in this way.

\begin{theorem}[Finite-versus-infinite parameter sets]\label{thm:parameters}
Fix an input representation $\rho:F_n\to\GL(V)$ with finite image.
\begin{enumerate}
\item If at least one irreducible constituent has at least two active
 generators, then
 \[
 \mathcal F(\rho)=\{\lambda\ne0,1:\MCDR_\lambda(\rho)(F_n)
                         \text{ is finite}\}
 \]
 is finite and effectively enumerable from exact input data.
\item If every irreducible constituent has at most one active
 generator and the input is nontrivial, then $\mathcal F(\rho)$ is
 exactly the set of all roots of unity different from $1$.
\item For trivial input, the output is zero for every parameter in
 the domain, so $\mathcal F(\rho)=\C^\times\setminus\{1\}$.
\end{enumerate}
\end{theorem}

The following proposition gives the order bound used in case (1). It
combines a count of relative conjugates with the length of the
indefinite open interval from Lemma~\ref{lem:gap}.
Write $\omega(q)$ for the number of distinct prime factors of $q$
and $\varphi$ for Euler's totient function, with $\omega(1)=0$.

\begin{proposition}[Relative conjugates and an effective order bound]\label{prop:order-bound}
\begin{enumerate}[label=\textup{(\roman*)}]
\item Let $M,q$ be positive integers, put $d=\gcd(M,q)$, and let
 $b\in\mathbb Z$ satisfy $\gcd(b,q)=1$.
 For an open interval $I\subset(0,1)$ of length $\delta$,
 let $\mathcal C$ count the integers $k$ satisfying
 $k/q\in I$, $k\equiv b\pmod d$, and $\gcd(k,q)=1$. Then
 \[
  \left|\mathcal C-\delta\frac{\varphi(q)}{\varphi(d)}\right|
  \le 2^{\omega(q)}.
 \]
\item Let $\rho:F_n\to\GL(V)$ have finite image, let $M$ be the
 exponent of $\rho(F_n)$, and let $U\subseteq V$ be an irreducible
 constituent with at least two active generators.
 Choose a nonempty open interval $I\subset(0,1)$ on which
 the canonical quotient form for $U$ is nondegenerate and indefinite,
 as in Lemma~\ref{lem:gap}, and let $\delta$ be its length.
 Every $\lambda\ne0,1$ for which
 $\MCDR_\lambda(\rho)(F_n)$ is finite is a root of unity of order
 $q>1$ satisfying
 \begin{equation}\label{eq:sieve}
  \frac{\varphi(q)}{2^{\omega(q)}}\le
  C,\qquad C=\frac{\varphi(M)}{\delta}.
 \end{equation}
 If $I$ is chosen by the finite procedure in
 Lemma~\ref{lem:gap}, then $\delta\ge1/M$ and $C\le M\varphi(M)$.
 Thus the universal bound $M\varphi(M)$ is always available.
 For each fixed $C>0$, the positive integers $q$ satisfying
 \eqref{eq:sieve} form a finite, effectively enumerable set.
\end{enumerate}
\end{proposition}
\begin{proof}
\textit{Part (i).}
Primes dividing $d$ cannot divide $k$ in this progression. Apply
inclusion-exclusion to the primes dividing $q$ but not $d$.
For a squarefree product $s$ of these primes, the additional condition
$s\mid k$ gives one progression modulo $ds$. Its count in the open
interval $qI$ differs from $\delta q/(ds)$ by at most one.
Summing the errors proves the bound. The main term is
\[
 \frac{\delta q}{d}\prod_{p\mid q,\ p\nmid d}(1-p^{-1})
 =\delta\frac{\varphi(q)}{\varphi(d)}.
\]

\smallskip\noindent\textit{Part (ii).}
The input is nontrivial, so Proposition~\ref{prop:nonroot} implies
that a finite-output parameter is a root of unity.
Write $\lambda=e^{2\pi ib/q}$ with $q>1$, $1\le b<q$, and
$\gcd(b,q)=1$. As $\sigma$ ranges over the field automorphisms of
$\overline{\Q}$ fixing $\Q(\zeta_M)$ pointwise, the values
$\sigma(\lambda)$ are precisely the elements of
\[
 \left\{e^{2\pi ik/q}\;\middle|\;
 1\le k<q,\ \gcd(k,q)=1,\ k\equiv b\pmod{\gcd(M,q)}\right\}.
\]
Indeed, this follows by the Chinese remainder theorem from exponents
$t\equiv1\pmod M$ in $(\mathbb Z/\lcm(M,q)\mathbb Z)^\times$.
Extend the corresponding automorphism to $\overline{\Q}$.
If $\MCDR_\lambda(\rho)$ has finite image, so does
$\MCDR_\lambda(U)$. Applying \eqref{eq:galois-finiteness} to $U$
shows that the output at every parameter $e^{2\pi ik/q}$ in this set
also has finite image. None of the corresponding numbers $k/q$ can
belong to $I$:
by Theorem~\ref{thm:free}, finite image would force the canonical
quotient form of $U$ to be definite, whereas it is indefinite
throughout this interval.

Part (i), with $\mathcal C=0$, therefore gives
\[
 \delta\frac{\varphi(q)}{\varphi(d)}\le2^{\omega(q)},
 \qquad d=\gcd(M,q).
\]
Since $\varphi(d)\le\varphi(M)$, this implies \eqref{eq:sieve}.
For the universal bound, choose $I$ by the finite procedure
in Lemma~\ref{lem:gap}. The product of the input matrices on $U$
has order dividing $M$. Its eigenangles
$\gamma^{(\nu)}$, $1\le\nu\le\dim U$, therefore have the form
$\gamma^{(\nu)}=b_\nu/M$ with $b_\nu\in\{0,1,\ldots,M-1\}$.
All the boundary points $1-\gamma^{(\nu)}$ with
$\gamma^{(\nu)}\ne0$, together with $0,1$, belong to
\[
 \left\{0,\frac1M,\frac2M,\ldots,1\right\}.
\]
The interval returned by Lemma~\ref{lem:gap} lies between consecutive
distinct boundary points, so its length satisfies $\delta\ge1/M$.
Consequently $C=\varphi(M)/\delta\le M\varphi(M)$.

Finally, the bound admits only finitely many $q$, effectively. Indeed,
\[
 \frac{\varphi(q)}{2^{\omega(q)}}
 =\prod_{p^e\parallel q}\frac{p^{e-1}(p-1)}2.
\]
Here $p^e\parallel q$ means $p^e\mid q$ and $p^{e+1}\nmid q$.
Every odd-prime factor on the right is at least one; the factor for
$2$ is at least $1/2$. Hence each odd prime divisor satisfies
$p\le4C+1$ and $p^{e-1}(p-1)\le4C$, while $2^e\mid q$ implies
$2^e\le4C$. These bounds give a finite search of prime powers,
after which \eqref{eq:sieve} selects the required integers $q$.
\end{proof}

\begin{proof}[Proof of Theorem~\ref{thm:parameters}]
For a nontrivial input, Proposition~\ref{prop:nonroot} restricts
finite-output parameters to roots of unity.
In case (1), apply Lemma~\ref{lem:gap} to a constituent with at least
two active generators. Proposition~\ref{prop:order-bound}\textup{(ii)} then gives
a finite, effectively enumerable set of candidate orders.
For each candidate $q>1$, test every
$\lambda=e^{2\pi ib/q}$ with $1\le b<q$ and $\gcd(b,q)=1$ using
the all-embedding definiteness criterion of Theorem~\ref{thm:free}.
These tests can be evaluated by Corollary~\ref{cor:angles}.
The finite input group, its exponent, algebraic irreducible
constituents, and their spectra are all computable from the
stipulated exact finite-image data. Thus this procedure enumerates
$\mathcal F(\rho)$ and terminates.
We do not estimate the computational complexity of this procedure.

In case (2), each constituent is one-dimensional: its image is
cyclic. A constituent with unique active scalar $a$ has zero output
when $\lambda a=1$, and otherwise a one-dimensional output on which
that generator acts by $\lambda a$ and all others act trivially.
Thus every root of unity gives finite image on every constituent.
A finite direct sum of such finite images is finite. Non-roots are
excluded as above. Case (3) follows from $K=V^{\oplus n}$.
\end{proof}

\begin{corollary}\label{cor:parameter-contrast}
Let $\rho:\Gamma\to\GL(V)$ have finite image, where
$\Gamma=F_n\rtimes B_n$, and put $\Gamma_{\rm pure}:=F_n\rtimes PB_n$.
For each $J\in\{F_n,\Gamma_{\rm pure},\Gamma\}$, define
\[
 \mathcal F_J(\rho):=
 \{\lambda\in\C\setminus\{0,1\}:T_\lambda(J)\text{ is finite}\}.
\]
Then
\[
 \mathcal F_\Gamma(\rho)=\mathcal F_{\Gamma_{\rm pure}}(\rho)
 =\mathcal F_{F_n}(\rho)=\mathcal F(\rho|_{F_n}),
\]
where the last set is defined in Theorem~\ref{thm:parameters}.
Thus the parameters giving finite image on $F_n\rtimes B_n$
are exactly those giving finite image on $F_n\rtimes PB_n$, and
these are also exactly those giving finite image on $F_n$.
The classification and effective enumeration in
Theorem~\ref{thm:parameters}, applied to $\rho|_{F_n}$,
therefore determine this common parameter set.

The corresponding set for the braid subgroup $B_n$ alone can be
larger: there is a finite-image input for which
$\mathcal F_\Gamma(\rho)=\{i\}$, although $T_\lambda(B_n)$ is finite
for every $\lambda\in\C\setminus\{0,1\}$.
\end{corollary}
\begin{proof}
Corollary~\ref{thm:full} gives
$\mathcal F_\Gamma(\rho)=\mathcal F_{F_n}(\rho)$.
Since $[\Gamma:\Gamma_{\rm pure}]=n!$, for every parameter we have
$[T_\lambda(\Gamma):T_\lambda(\Gamma_{\rm pure})]\le n!$.
Hence $T_\lambda(\Gamma)$ is finite if and only if
$T_\lambda(\Gamma_{\rm pure})$ is finite, proving the remaining group comparison.
The identity $T_\lambda|_{F_n}=\MCDR_\lambda(\rho|_{F_n})$
identifies the common set with $\mathcal F(\rho|_{F_n})$.
The scalar input in Appendix~\ref{sec:scalar-parameters} supplies
the final example.
\end{proof}

\begin{remark}
The assumption in case (1) of Theorem~\ref{thm:parameters} is that
at least two generators act nontrivially on some irreducible input
constituent. A condition on the total output dimension alone would
not suffice: direct sums of one-dimensional inputs with at most
one active generator can have total output dimension greater than
one and still give finite image for every root of unity $\lambda\ne1$.
The theorem also allows the output dimension to drop to one at
exceptional parameters; it does not require dimension at least two at every parameter.
For inputs defined on $\Gamma=F_n\rtimes B_n$,
Section~\ref{sec:resonance} studies the structure of the braid
representation at exceptional parameters.
\end{remark}

\section{Algebraic solutions of regular-singular systems}\label{sec:algebraic}

The finite-image criteria now give an analytic consequence: they
characterize when all horizontal solutions of the corresponding
regular-singular systems are algebraic. We first recall the classical
finite-monodromy criterion, then apply the free-group and full-image
results to one-variable systems and KZ-type equations, respectively.

Let $\nabla$ be an algebraic flat connection of rank $d$ on a
smooth connected complex algebraic variety $X$.
A \emph{horizontal solution} is a local analytic section $s$
satisfying $\nabla s=0$. In an algebraic local frame, write
$\nabla=d-\Omega$; the horizontal equation is then $ds=\Omega s$.
On a sufficiently small simply connected analytic open set, a
\emph{horizontal basis} consists of $d$ horizontal solutions whose
values form a basis in each fiber. Writing these solutions as columns
in the chosen frame gives a \emph{fundamental solution matrix}
\[
 Y=(s_1\ \cdots\ s_d),\qquad dY=\Omega Y,\qquad \det Y\ne0.
\]
Every horizontal solution on this open set is $s=Yc$ for a constant
vector $c\in\C^d$. Algebraicity of all horizontal solutions means
that the entries of $Y$ in an algebraic local frame belong to a
finite algebraic extension of the function field $\C(X)$.

Regular singularity is understood to include behavior at infinity.
We use the following classical consequence of the regular-singular
Riemann--Hilbert correspondence
\cite[Chapter~II, Theorem~5.9]{Deligne}.

\begin{proposition}[Classical finite-monodromy criterion]\label{prop:algebraicity}
A regular-singular algebraic flat connection has finite monodromy if and
only if all its horizontal solutions are algebraic.
\end{proposition}

For a punctured line, Proposition~\ref{prop:algebraicity} and
Theorem~\ref{thm:free} give a complete criterion for algebraicity of all
solutions of the regular-singular system realizing the convolved tuple.
This one-variable application requires no braid action.

Let
\[
 X_{n+1}=\{(z_0,\ldots,z_n)\in\C^{n+1}:z_i\ne z_j\text{ for }i\ne j\}.
\]
Following \cite[Definition 13]{KLM}, a \emph{KZ-type system} on
$X_{n+1}$ is a system
\begin{equation}\label{eq:kz-system}
 du=\left(\sum_{0\le i<j\le n}A_{ij}\,d\log(z_i-z_j)\right)u,
\end{equation}
where $u$ is a column of $d$ functions and $A_{ij}=A_{ji}\in M_d(\C)$
are constant matrices satisfying
\[
 [A_{ij},A_{k\ell}]=0\quad(\{i,j\}\cap\{k,\ell\}=\varnothing),
 \qquad [A_{ij},A_{ik}+A_{jk}]=0\quad(i,j,k\text{ distinct}).
\]
All indices range from $0$ to $n$, with distinct indices in each pair;
$[A,B]=AB-BA$. These conditions ensure flatness. The two-index residue
matrices $A_{ij}$ in this subsection are distinct from the convolution
operators $A_i$ of \eqref{eq:Ai}.

With compatible base points and generators,
$\pi_1(X_{n+1})\simeq PB_{n+1}\simeq F_n\rtimes PB_n$.
Haraoka constructs the multiplicative middle convolution for KZ
equations in \cite{Haraoka2020}. Its relation to the restriction of KLM
to this group is established in \cite{KLM}. We state the consequence with this
identification as an explicit assumption. In particular, any
rank-one character twist arising from the chosen normalization
(multiplying each monodromy operator by a character value) must be
included when checking this identification.

\begin{corollary}[Algebraicity for the corresponding KZ-type output]
\label{cor:kz}
Let $E$ be a number field, let
$\rho:\Gamma\to\GL_E(V)$ have finite full image, and let
$\nabla_\lambda$ be a regular-singular algebraic flat connection on
the ordered configuration space
$X_{n+1}=\{(z_0,\ldots,z_n)\in\C^{n+1}:z_i\ne z_j\ (i\ne j)\}$
whose monodromy is isomorphic to
$T_\lambda|_{F_n\rtimes PB_n}$ in the conventions of this paper.
Then all horizontal solutions of $\nabla_\lambda$ are algebraic if and
only if $T_\lambda(F_n)$ is finite. If the input restriction to $F_n$
is nontrivial, these conditions are equivalent to the root-of-unity
and componentwise eigenangle conditions of
Corollary~\ref{cor:full-spectral}.
\end{corollary}
\begin{proof}
Apply Proposition~\ref{prop:algebraicity} and
Corollary~\ref{cor:full-spectral}. If the restriction $\rho|_{F_n}$ is trivial, the
KLM output is zero and the assertion has the stipulated trivial meaning.
\end{proof}

If such connections are specified for all $\lambda$ in a parameter
domain $\mathcal D\subseteq\C\setminus\{0,1\}$, their algebraicity
locus is exactly $\mathcal D\cap\mathcal F(\rho|_{F_n})$.
Theorem~\ref{thm:parameters} therefore classifies that locus, and
Corollary~\ref{cor:parameter-contrast} identifies it using the pure braid
restriction. This statement presupposes the family and the monodromy
identification in Corollary~\ref{cor:kz}.

This proves algebraicity for the output of the specified monodromy
transformation; it does not supply polynomial equations for a fundamental
solution matrix. A presentation of those functions requires the analytic
integral transform or a further algebraic calculation.

\begin{remark}[Iterating the construction]
If the full output at one step is finite, it is an admissible finite-image
input for another KLM step. The same criterion may be applied at each
step. An intermediate output which fails the criterion is not a finite
input for the next use of the theorem. No unrestricted assertion about
arbitrary iterated middle convolutions follows.
\end{remark}

\section{Braid representations when \texorpdfstring{$L_\lambda\ne0$}{L(lambda) is nonzero}}\label{sec:resonance}

In this section, we consider a complex input representation
$\rho:\Gamma\to\GL(V)$ whose image is not assumed to be finite.
Whenever a result requires the image $\rho(B_n)$ of the braid subgroup
to be finite, we state this additional assumption explicitly.
\paragraph{Exceptional parameters.}
Write $L_\lambda$ for the subspace $L$ of Section~\ref{sec:free},
emphasizing its dependence on $\lambda$. A parameter
$\lambda\in\C\setminus\{0,1\}$ is \emph{exceptional} if
$\lambda^{-1}\in\Spec(P_n)$ and \emph{nonexceptional} otherwise.
The finite set of exceptional parameters is
\begin{equation}\label{eq:exceptional-set}
 \mathcal E_\rho=\{\mu^{-1}:\mu\in\Spec(P_n)\}\setminus\{1\}.
\end{equation}
Equation~\eqref{eq:Lchain} below gives the equivalent conditions
\[
 \lambda\in\mathcal E_\rho
 \quad\Longleftrightarrow\quad
 \ker(\lambda P_n-I_N)\ne0
 \quad\Longleftrightarrow\quad L_\lambda\ne0.
\]
Thus an exceptional parameter is one at which an additional nonzero
subspace is removed in forming the KLM quotient. Since
$K\cap L_\lambda=0$ for $\lambda\ne1$, the output dimension drops
by $\dim L_\lambda$ from $\dim(V^{\oplus n}/K)$. This definition depends only on the product
$P_n=g_1\cdots g_n$, so we also use it for $F_n$-inputs and for their
irreducible constituents, using the product on the chosen constituent.

The operators $S_i$ satisfy the braid relations by the Long--Moody
construction \cite{KLM} and preserve $K$ (also checked in the proof
of Lemma~\ref{lem:boundary-multiplication} below). We may therefore
define the generic braid representation
\[
 U_0=V^{\oplus n}/K,\qquad T_{\rm gen}:B_n\to\GL(U_0).
\]
Here $K=\bigoplus_{j=1}^n\ker(g_j-I_N)$, with the $j$-th
kernel embedded in the $j$-th coordinate of $V^{\oplus n}$.
Thus $U_0$ is the intermediate space obtained by quotienting out
only the local fixed subspaces. The KLM output space is obtained
by taking the further quotient
\[
 Q_\lambda=V^{\oplus n}/(K+L_\lambda)
 \cong U_0\big/\bigl((K+L_\lambda)/K\bigr).
\]
At a nonexceptional parameter, $L_\lambda=0$, so $Q_\lambda$ is
canonically identified with $U_0$. The subscript $0$ in $U_0$
is a label for this parameter-independent space; it does not mean
specialization at $\lambda=0$.

The operators $S_i$ in \eqref{eq:Si} induce $T_{\rm gen}$ and are
independent of $\lambda$. This notation does not refer to a
representation over a field of rational functions.

For two representations $\rho_a:B_n\to\GL(W_a)$ ($a=1,2$),
a linear map $f:W_1\to W_2$ is called \emph{braid-equivariant}
(or a \emph{homomorphism of $B_n$-representations}) if
\[
 f\bigl(\rho_1(b)w\bigr)=\rho_2(b)f(w)
 \qquad\text{for all }b\in B_n\text{ and }w\in W_1.
\]
Thus applying a braid before or after the map gives the same result.
It suffices to check this identity for the standard generators of $B_n$.

\subsection{An equivariant embedding and the multiplication map}

The purpose of this subsection is to identify the braid-invariant subspace
by which $U_0=V^{\oplus n}/K$ is further quotiented when $L_\lambda\ne0$.
We use the braid-equivariant embedding $F:V\hookrightarrow V^{\oplus n}$
and the multiplication map $D:V^{\oplus n}\to V$ recalled below.
Here $V$ carries the restriction $\rho|_{B_n}$ of the input representation,
and $V^{\oplus n}$ carries the intermediate Long--Moody braid representation.
The braid-invariant subspace $F(V)$ is therefore isomorphic to
$\rho|_{B_n}$. Both maps are described explicitly in
Lemma~\ref{lem:boundary-multiplication}; the map $D$ is not in general
an inverse of $F$.

These maps belong to the established reduced Long--Moody construction.
A copy of the input braid representation appears in Long's reduction
\cite[Theorem 2.11]{Long1994}; see also \cite[Section 6]{BigelowTian}.
For the multiplication map in the augmentation-ideal description, see
\cite[Section 1.2.2 and Proposition 1.10]{HN}; for the invariant subspaces used in the KLM quotient, see
\cite[Section 3 and Proposition 3.1]{HN}.

Here we first pass to the quotient by the local fixed subspaces.
At exceptional parameters, the additional subspace removed by the KLM construction
comes from a particular eigenspace of the product of the input free
generators. The identities below describe this subspace and will let
us construct a braid-equivariant retraction under the assumptions of
Proposition~\ref{prop:resonance-splitting}. This leads to the splitting
at exceptional parameters, the comparison of image groups in
Proposition~\ref{prop:image-kernel}, and the parameter-independence
of braid-image finiteness in Theorem~\ref{thm:braid-dichotomy}.

\begin{lemma}\label{lem:boundary-multiplication}
Put $P_n:=g_1\cdots g_n$ and define
\begin{equation}\label{eq:boundary-embedding}
 F:V\longrightarrow V^{\oplus n},\qquad
 F(v)=(g_2\cdots g_nv,\ g_3\cdots g_nv,\ldots,g_nv,v).
\end{equation}
Recall the multiplication map of \eqref{eq:D-definition}:
\[
 D:V^{\oplus n}\longrightarrow V,\qquad
 D(v_1,\ldots,v_n)=\sum_{j=1}^n(g_j-I_N)v_j.
\]
The map $F$ is injective, and its image $F(V)$
is the subspace of tuples satisfying $v_j=g_{j+1}v_{j+1}$
for $1\le j<n$. We have
\begin{equation}\label{eq:FD-equivariance}
 S_iF=Fs_i,\qquad DS_i=s_iD,\qquad s_i P_n=P_n s_i,
\end{equation}
and
\begin{equation}\label{eq:DF-product}
 DF=P_n-I_N,\qquad D(K)=0.
\end{equation}
Consequently $D$ induces the braid-equivariant map
\[
 \bar D:U_0=V^{\oplus n}/K\longrightarrow V,\qquad
 \bar D([v])=D(v)=\sum_{j=1}^n(g_j-I_N)v_j,
\]
where $v=(v_1,\ldots,v_n)$ and $[v]=v+K$. This is independent of
the representative because $D(v+k)=D(v)$ for every $k\in K$,
by $D(K)=0$. Here the bar denotes the induced map on the quotient,
not complex conjugation. For \(\lambda\ne1\),
\begin{equation}\label{eq:Lchain}
 \begin{split}
 L_\lambda
 &=\{(v_1,\ldots,v_n):
       v_j=g_{j+1}v_{j+1}\ (j<n),\
       v_n=\lambda P_nv_n\}\\
 &=F\bigl(\ker(\lambda P_n-I_N)\bigr),
 \end{split}
\end{equation}
and \(K\cap L_\lambda=0\).
\end{lemma}

\begin{proof}
Applying $\rho$ to \eqref{eq:artin-action} gives
\[
 s_ig_i=g_{i+1}s_i,\qquad
 s_i(g_ig_{i+1})=(g_ig_{i+1})s_i,\qquad
 s_ig_j=g_js_i\quad(j\notin\{i,i+1\}).
\]
The last two identities give \(s_i P_n=P_n s_i\).
Put \(a=g_i\), \(b=g_{i+1}\), \(s=s_i\), and
\(C=g_{i+2}\cdots g_n\). The two relevant coordinates of \(F(v)\)
are \((bCv,Cv)\). The block in \eqref{eq:Si} sends these to
\((aCv,Cv)\), and multiplication by \(s\) gives
\((bCs v,Cs v)\). All other coordinates follow from the same
commutation identities. This proves \(S_iF=Fs_i\).

For the second identity, use
\(s^{-1}as=aba^{-1}\) and \(s^{-1}bs=a\).
After factoring out \(s\) from \(DS_i\), the coefficient of
coordinate \(i\) is \(a-I_N\), while that of coordinate \(i+1\) is
\[
 (aba^{-1}-I_N)a+(a-I_N)(I_N-b)=b-I_N.
\]
The other coefficients are unchanged, so \(DS_i=s_iD\).
The telescoping sum
\[
 \sum_{j=1}^n(g_j-I_N)g_{j+1}\cdots g_n=P_n-I_N
\]
proves \(DF=P_n-I_N\), and \(D(K)=0\) follows coordinatewise.

For completeness, \(K\) is braid-invariant: if \(av_i=v_i\)
and \(bv_{i+1}=v_{i+1}\), the new \(i+1\)-st coordinate is
\(sv_i\), which is fixed by \(b\), and the new \(i\)-th coordinate
\(sa v_{i+1}\) satisfies
\((a-I_N)sa v_{i+1}=sa(b-I_N)v_{i+1}=0\).
The other coordinates are treated by commutation.

Let \(e_i(v)\) denote the nonzero block row of \((A_i-I_{nN})v\).
From \eqref{eq:Ai},
\[
 e_i(v)-e_{i+1}(v)
   =(\lambda-1)(v_i-g_{i+1}v_{i+1}).
\]
If $v\in L_\lambda$ and $\lambda\ne1$, then $e_i(v)=0$ for
$1\le i\le n$. The difference identity above gives
$v_j=g_{j+1}v_{j+1}$ for $1\le j<n$, hence $v=F(v_n)$.
The last equation is
$e_n(F(v_n))=(\lambda P_n-I_N)v_n=0$.
Thus $v$ satisfies the conditions on the right-hand side of
\eqref{eq:Lchain}.

Conversely, let $v=(v_1,\ldots,v_n)$ satisfy those conditions.
The coordinate relations give $v=F(v_n)$ and
$e_i(v)-e_{i+1}(v)=0$ for $1\le i<n$.
The condition $v_n=\lambda P_nv_n$ gives
$e_n(v)=(\lambda P_n-I_N)v_n=0$. Therefore $e_i(v)=0$ for every
$1\le i\le n$, so $v\in\bigcap_{i=1}^n\ker(A_i-I_{nN})=L_\lambda$.
This proves the reverse inclusion and hence \eqref{eq:Lchain}.
Finally, if
\(F(v)\in K\) and \(P_nv=\lambda^{-1}v\), then
\[
 0=DF(v)=(P_n-I_N)v=(\lambda^{-1}-1)v.
\]
Hence \(v=0\), proving the intersection assertion.
\end{proof}

\begin{proposition}[Parameter reduction]\label{prop:parameter}
For any fixed complex input \(\rho\), the following hold.
\begin{enumerate}[label=(\roman*)]
 \item If \(\lambda\notin\mathcal E_\rho\) and \(\lambda\ne0,1\),
 then \(T_\lambda|_{B_n}\cong T_{\rm gen}\).
 \item For every \(\lambda\ne0,1\), the representation
 \(T_\lambda|_{B_n}\) is a quotient of \(T_{\rm gen}\).
\end{enumerate}
\end{proposition}

\begin{proof}
Both \(K\) and the braid operators are independent of \(\lambda\).
Equation~\eqref{eq:Lchain} gives \(L_\lambda=0\) outside
\(\mathcal E_\rho\), proving (i). The natural quotient map
\[
 U_0\longrightarrow V^{\oplus n}/(K+L_\lambda)
\]
is braid-equivariant, proving (ii).
\end{proof}

\subsection{Splitting at exceptional parameters}

\begin{proposition}[Equivariant splitting at exceptional parameters]
\label{prop:resonance-splitting}
Let \(\mu\in\Spec(P_n)\setminus\{1\}\), and put
\(V_\mu=\ker(P_n-\mu I_N)\).
By \eqref{eq:FD-equivariance}, $P_n$ commutes with $\rho(b)$ for
every $b\in B_n$, so $V_\mu$ is invariant under $\rho(B_n)$.
We denote the resulting representation on this invariant subspace by
\[
 \rho|_{B_n,V_\mu}:B_n\longrightarrow\GL(V_\mu),\qquad
 \bigl(\rho|_{B_n,V_\mu}(b)\bigr)(v)=\rho(b)v
 \quad(b\in B_n,\ v\in V_\mu).
\]
Thus this notation means first restricting the group to $B_n$ and
then restricting its action to the invariant subspace $V_\mu$.
Suppose that the inclusion \(V_\mu\subset V\) admits a
braid-equivariant retraction \(p_\mu:V\to V_\mu\) for the input
representation. Explicitly, $p_\mu$ is a linear map satisfying
\[
 p_\mu(v)=v\quad(v\in V_\mu),\qquad
 p_\mu\bigl(\rho(b)v\bigr)
 =\bigl(\rho|_{B_n,V_\mu}(b)\bigr)p_\mu(v)
 \quad(b\in B_n,\ v\in V).
\]
Thus $p_\mu$ restricts to the identity on $V_\mu$ and commutes with
the braid action. Then
\begin{equation}\label{eq:resonance-splitting}
 T_{\rm gen}\cong
 \rho|_{B_n,V_\mu}\oplus T_{\mu^{-1}}|_{B_n}.
\end{equation}
In particular, either of the following sufficient conditions guarantees
the existence of $p_\mu$:
\begin{enumerate}[label=(\roman*)]
 \item \(P_n\) is semisimple;
 \item \(\rho(B_n)\) has finite image.
\end{enumerate}
In particular, case (ii) does not require \(P_n\) to be semisimple.
Whenever \(\rho|_{B_n,V_\mu}\) has finite image, the generic and
KLM braid representations at exceptional parameters in
\eqref{eq:resonance-splitting} have finite image simultaneously.
\end{proposition}

\begin{proof}
Given the braid-equivariant retraction \(p_\mu\), define
\[
 j_\mu:V_\mu\longrightarrow U_0,\qquad j_\mu(v)=[F(v)],
 \qquad
 r_\mu:U_0\longrightarrow V_\mu,\qquad
 r_\mu=\frac1{\mu-1}p_\mu\bar D.
\]
Both maps are braid-equivariant: this follows for $j_\mu$ from
$S_iF=Fs_i$ in \eqref{eq:FD-equivariance} and the equivariance of the
quotient map, and for $r_\mu$ from the equivariance of $\bar D$ and
$p_\mu$. Moreover,
\[
 r_\mu j_\mu
 =\frac1{\mu-1}p_\mu(P_n-I_N)|_{V_\mu}
 =I_{V_\mu}.
\]
Here it is enough that \(p_\mu|_{V_\mu}=I_{V_\mu}\); commutation of
\(p_\mu\) with \(P_n\) is not needed.
It follows that
\[
 U_0=j_\mu(V_\mu)\oplus\ker r_\mu.
\]
By \eqref{eq:Lchain}, \(L_{\mu^{-1}}=F(V_\mu)\).
Its image in \(U_0\) is therefore \(j_\mu(V_\mu)\), and the
natural projection onto the KLM quotient restricts to an isomorphism
\[
 \ker r_\mu\xrightarrow{\ \sim\ }
 V^{\oplus n}/(K+L_{\mu^{-1}}).
\]
This proves \eqref{eq:resonance-splitting}.

If \(P_n\) is semisimple, take its spectral projection onto \(V_\mu\).
It is a polynomial in \(P_n\), hence commutes with every \(s_i\).
If instead $\rho(B_n)$ is finite, start with any linear
retraction $p_0:V\to V_\mu$ and set
\[
 p_\mu=\frac1{|\rho(B_n)|}\sum_{g\in\rho(B_n)}g^{-1}p_0g.
\]
Since $V_\mu$ is invariant under $\rho(B_n)$, this map has image in
$V_\mu$, restricts to the identity there, and commutes with $\rho(B_n)$.
Thus it is the required retraction without any assumption on the
Jordan form of \(P_n\).

Finally, a direct sum with a finite-image representation has
finite image if and only if its other summand does.
\end{proof}

\begin{proposition}[The change of image at exceptional parameters]\label{prop:image-kernel}
In the splitting situation of Proposition~\ref{prop:resonance-splitting},
write $H_\mu=(\rho|_{B_n,V_\mu})(B_n)$,
$G_{\rm gen}=T_{\rm gen}(B_n)$, and
$G_\mu=T_{\mu^{-1}}(B_n)$. There is an exact sequence
\[
 1\longrightarrow \mathcal N_\mu\longrightarrow G_{\rm gen}
   \longrightarrow G_\mu\longrightarrow1,
 \qquad \mathcal N_\mu\hookrightarrow H_\mu.
\]
More precisely,
\[
 \mathcal N_\mu\ne\{1\}
 \quad\Longleftrightarrow\quad
 \ker(T_{\mu^{-1}}|_{B_n})
 \not\subseteq\ker(\rho|_{B_n,V_\mu}).
\]
Equivalently, some braid acts as the identity on the quotient
$Q_{\mu^{-1}}$ but acts nontrivially on the summand corresponding to
$V_\mu$.
If $H_\mu$ is finite, the surjective homomorphism has finite kernel. In the finite
case, $|G_{\rm gen}|=|\mathcal N_\mu||G_\mu|$ and $|\mathcal N_\mu|$ divides $|H_\mu|$.
\end{proposition}
\begin{proof}
Under the splitting, the generic image is the subgroup of
$H_\mu\times G_\mu$ consisting of pairs
$(\rho|_{B_n,V_\mu}(b),T_{\mu^{-1}}(b))$ for $b\in B_n$.
Projection to the second factor is onto. On its kernel, projection to
the first factor is injective, identifying
\[
 \mathcal N_\mu\simeq\rho|_{B_n,V_\mu}\bigl(\ker(T_{\mu^{-1}}|_{B_n})\bigr).
\]
This identification gives the stated criterion for a nontrivial kernel.
The group-order assertions follow from elementary group theory.
\end{proof}

\begin{remark}
This group extension is not asserted to split. A direct-sum decomposition
of representations embeds the image into a product; it need not identify
the image with that product.
\end{remark}

\begin{theorem}[Parameter independence for the braid restriction]
\label{thm:braid-dichotomy}
If $\rho(B_n)$ is finite, then
\[
 \{\lambda\in\C\setminus\{0,1\}:T_\lambda(B_n)\text{ is finite}\}
 =\C\setminus\{0,1\}\quad\text{or}\quad\varnothing.
\]
The first alternative holds exactly when $T_{\rm gen}(B_n)$ is finite.
\end{theorem}
\begin{proof}
For $\lambda\notin\mathcal E_\rho$ (with $\lambda\ne0,1$),
Proposition~\ref{prop:parameter}(i) gives
$T_\lambda|_{B_n}\cong T_{\rm gen}$.
At $\lambda=\mu^{-1}$, Proposition~\ref{prop:resonance-splitting}
identifies $T_{\rm gen}$ with the sum of $T_{\mu^{-1}}|_{B_n}$
and a finite-image input summand. Their finiteness is therefore equivalent; equivalently use the
surjective homomorphism with finite kernel in Proposition~\ref{prop:image-kernel}.
\end{proof}

Examples~\ref{ex:resonance-kernel}--\ref{ex:braid-assumption} in
Appendix~\ref{sec:braid-examples} illustrate the preceding results and
the role of the finite braid-input assumption.

The splitting is a statement about $B_n$-representations. Its summands
need not be invariant under the output free generators. In particular,
Theorem~\ref{thm:braid-dichotomy} is consistent with the root-of-unity
restriction for finiteness of the full image in Corollary~\ref{cor:full-spectral}.

\subsection{Projective kernels at exceptional parameters}
\begin{proposition}\label{p:kernel}
Let $G\subset\GL(V_0\oplus W)$ be the image of a representation
preserving both summands, with $\dim W>0$. Define the restriction map
and its image by
\[
 \pi_W:G\longrightarrow\GL(W),\qquad
 \pi_W(g)=g|_W,\qquad G_W:=\pi_W(G).
\]
Thus $\pi_W:G\twoheadrightarrow G_W$ is the surjective homomorphism
induced by the action on $W$. Set $\mathcal N=\ker\pi_W$,
\[
 Z=G\cap\C^\times I_{V_0\oplus W},\qquad
 Z_W=G_W\cap\C^\times I_W.
\]
Write $\PG G$ and $\PG G_W$ for the images of $G$ and $G_W$ in
$\operatorname{PGL}(V_0\oplus W)$ and $\operatorname{PGL}(W)$,
respectively. Equivalently, these are the quotient groups
\[
 \PG G=G/Z,\qquad \PG G_W=G_W/Z_W,
\]
where scalar matrices belonging to each group are identified with the
identity. The map between the projective images is induced by
restriction: $gZ\mapsto\pi_W(g)Z_W$.
Then $\pi_W|_Z$ is injective, and there are exact sequences
\[
 1\longrightarrow K_{\rm proj}\longrightarrow \PG G
   \longrightarrow\PG G_W\longrightarrow1,
 \qquad K_{\rm proj}=\pi_W^{-1}(Z_W)/Z,
\]
\[
 1\longrightarrow \mathcal N\longrightarrow K_{\rm proj}
   \longrightarrow Z_W/\pi_W(Z)\longrightarrow1.
\]
In the finite case,
\[
 \frac{|\PG G|}{|\PG G_W|}
 =|\mathcal N|\frac{|Z_W|}{|Z|}.
\]
\end{proposition}
\begin{proof}
A scalar on the full space restricts to the same scalar on $W$. Since
$W\ne0$, this restriction is injective on scalar matrices. Projection
therefore induces a surjection of projective groups with the indicated
kernel. Also $\mathcal N\cap Z=1$, so $\mathcal N$ embeds in that kernel. The map induced
by $\pi_W$ from the kernel to $Z_W/\pi_W(Z)$ is onto and has precisely the image
of $\mathcal N$ as kernel. Counting elements gives the formula.
\end{proof}
For Proposition~\ref{prop:resonance-splitting}, provided the
quotient $Q_{\mu^{-1}}$ is nonzero, take
$V_0=j_\mu(V_\mu)$, $W=\ker r_\mu\simeq Q_{\mu^{-1}}$, and
$G=G_{\rm gen}$. Then $G_W\simeq G_\mu$ and the kernel $\mathcal N$
identifies with $\mathcal N_\mu$ of Proposition~\ref{prop:image-kernel}.

\subsection{Equivalence of linear and projective finiteness}
We now relate finiteness of the ordinary and projective images
for finite full input, using the projective-image notation introduced in
Proposition~\ref{p:kernel}.

\begin{lemma}\label{lem:projective}
Let $G\subset\GL_d(\C)$ be a subgroup, with $d>0$. Its determinant image is
\[
 \det(G):=\{\det(g):g\in G\}\subset\C^\times.
\]
Suppose this multiplicative group is finite of order $c_0$; that is,
the determinants of the matrices in $G$ take exactly $c_0$ distinct values. Then $G\cap\C^\times I_d$ has at most $dc_0$ elements. In particular,
$G$ is finite if and only if its projective image $\mathbb PG$ is finite.
\end{lemma}
\begin{proof}
For $zI_d\in G$ one has $z^d\in\det(G)$, hence $z^{dc_0}=1$.
The kernel of the map to the projective image is thus finite.
\end{proof}

\begin{proposition}\label{prop:projective-klm}
For finite full input $\rho:\Gamma\to\GL(V)$, the determinant image
of every $T_\lambda|_{B_n}$, $\lambda\ne0,1$, is finite. Consequently
finiteness of the linear braid image is equivalent to finiteness of
its projective image.
\end{proposition}
\begin{proof}
The determinant before taking the quotient is
\[
 \det S_i=(-1)^{\dim V}\det(s_i)^n\det(g_i),
\]
a root of unity. On $K$, the action is the restriction of block-monomial
matrices whose nonzero blocks belong to the finite input group. Thus
its image on $K$ is finite. On
$L_\lambda=F(\ker(\lambda P_n-I_N))$, the relation $S_iF=Fs_i$ identifies
the braid action with an input subrepresentation, so this image too is
finite. Since $K\cap L_\lambda=0$, multiplicativity of determinant on
a subspace and quotient gives
\[
 \det T_\lambda(\sigma_i)
 =\frac{\det S_i}{\det(S_i|_K)\det(S_i|_{L_\lambda})}.
\]
Each quotient is a root of unity, and the finitely many such roots
generate a finite determinant group. Apply Lemma~\ref{lem:projective}.
Zero-dimensional outputs have trivial image.
\end{proof}
\begin{remark}[Parameter independence of projective finiteness]\label{rem:projective-parameter-independence}
For finite full input, Proposition~\ref{prop:projective-klm} and
Theorem~\ref{thm:braid-dichotomy} imply that finiteness of the
projective braid image is independent of
$\lambda\in\C\setminus\{0,1\}$.
\end{remark}

\paragraph{Free, pure, and full images.}
\begin{proposition}\label{p:full}
Suppose $\rho:\Gamma\to\GL(V)$ has finite full image. For every
$\lambda\in\C\setminus\{0,1\}$, the following conditions are equivalent:
\[
\begin{array}{ll}
 T_\lambda(F_n)\text{ is finite}, &
 \PG T_\lambda(F_n)\text{ is finite},\\
 T_\lambda(\Gamma_{\rm pure})\text{ is finite}, &
 \PG T_\lambda(\Gamma_{\rm pure})\text{ is finite},\\
 T_\lambda(\Gamma)\text{ is finite}, &
 \PG T_\lambda(\Gamma)\text{ is finite}.
\end{array}
\]
Here $\Gamma_{\rm pure}=F_n\rtimes PB_n$.
Zero-dimensional outputs are understood to have trivial ordinary and
projective image. If the restriction $\rho|_{F_n}$ is nontrivial, any of these conditions
forces $\lambda$ to be a root of unity.
\end{proposition}
\begin{proof}
If all $g_j=I_N$, the output is zero. Otherwise, the Artin relations \eqref{eq:artin-action} give
$s_i g_i s_i^{-1}=g_{i+1}$, where $s_i=\rho(\sigma_i)$.
Thus the matrices $g_1,\ldots,g_n$ are conjugate to one another by elements
of $\rho(B_n)$, and $r_0=\rank(g_j-I_N)>0$ is independent of $j$.
Suppose first that $\lambda$ is not a root of unity. The finite order of
$P_n=g_1\cdots g_n$ gives $L_\lambda=0$, so
\[
 Q_\lambda=V^{\oplus n}/K,\qquad \dim Q_\lambda=nr_0.
\]
For an induced free generator, the image of $T_\lambda(x_i)-I_{nr_0}$ lies in the $i$th
coordinate quotient and has dimension at most $r_0$. Since $nr_0>r_0$, the
induced generator has eigenvalue $1$. Also, $g_i$ has an eigenvalue
$\mu\ne1$, a root of unity. An associated eigenvector supported in
coordinate $i$ survives modulo $K$ and is an eigenvector of $A_i$ with
eigenvalue $\lambda\mu$. Consequently $A_i$ has infinite projective order:
a scalar power would force $(\lambda\mu)^k=1$ for some $k>0$.
Both projective images are therefore infinite.

Now suppose $\lambda$ is a root of unity. On $K$, $A_i$ acts by $\lambda$
on the $i$th local fixed space and by the identity on the other coordinates;
it fixes $L_\lambda$ pointwise. Since $K\cap L_\lambda=0$, multiplicativity
of determinants gives
\[
 \det T_\lambda(x_i)=\lambda^{r_0}\det g_i.
\]
All these determinants are roots of unity. The braid determinants are
roots of unity by Proposition~\ref{prop:projective-klm}. The determinant images
of both generated groups are finite. Lemma~\ref{lem:projective} therefore
identifies ordinary and projective finiteness for each group. Corollary~\ref{thm:full}
identifies ordinary free and full finiteness. Finally, the inclusions
\[
 \begin{aligned}
 T_\lambda(F_n)&\subseteq T_\lambda(\Gamma_{\rm pure})
                 \subseteq T_\lambda(\Gamma),\\
 \PG T_\lambda(F_n)&\subseteq \PG T_\lambda(\Gamma_{\rm pure})
                 \subseteq \PG T_\lambda(\Gamma)
 \end{aligned}
\]
show that the two pure-image conditions are equivalent to these four
conditions as well.
\end{proof}
See Example~\ref{ex:free-projective} for the necessity of the braid extension.

\section{Discussion}\label{sec:discussion}

\begin{remark}[Comparison with parabolic cohomology]\label{rem:parabolic-comparison}
With the notation of Corollary~\ref{cor:angles}, augment the tuple by
$\tau(\lambda)I_N$ and $(\tau(\lambda)P_n)^{-1}$.
Its product is $I_N$, and $\tau(\lambda)\ne1$ eliminates common fixed
vectors. The sums of eigenangles of this augmented tuple and its
inverses are respectively $m(l)+N-r(l)$ and $R+N-m(l)$, counting the
inverse angle of $0$ as $0$. Theorem~2.3 of \cite{DW} therefore gives
$(m(l)-r(l),R-m(l))$. This is the reverse of the signature
of the canonical KLM form in Corollary~\ref{cor:angles}. We use only
this numerical agreement and do not claim that the canonical KLM form
is identified with the parabolic-cohomology pairing of \cite{DW}
under the Haraoka--Long correspondence.
\end{remark}

Unitary representations of braid groups also provide a mathematical
framework for quantum gates in topological quantum computation, where
braiding non-Abelian anyons acts on a computational state space
\cite{DRW2016,Nayak2008}. Finite-image criteria identify when the
reachable operations form a finite group. The braid-group representations
arising from Ising anyons provide familiar examples with finite image;
see \cite{AGW}.
For the stabilizer description of Clifford operations, see \cite{Gottesman1998}.
The unitary representation families and quantum-gate applications
are treated in a follow-up manuscript in preparation. They require a
positive definite invariant form, but need not have finite image;
using free generators in addition to braid operations requires
additional implementations. These applications are separate from the
finite-image criteria proved here.

\section*{Acknowledgements}
The author thanks Rapha\"el Belliard, Yoshiaki Goto, Toshio Oshima,
and Greyson Potter for helpful discussions.
This work was supported by JST SPRING, Grant Number JPMJSP2109.

\section*{Statements and declarations}
\paragraph{Use of AI-assisted tools.}
ChatGPT (OpenAI) was used to assist with drafting and revision,
literature comparison, exploring and checking proof arguments, and
preparing symbolic and exact-arithmetic checks. Claude (Anthropic)
was used for simulated referee reports and auxiliary checks.
AI-assisted checks supplement the mathematical arguments and do not
replace their proofs. The author is responsible for the mathematical
content, references, and final presentation.

\appendix

\section{Examples ordered by the results they illustrate}\label{sec:examples}
The examples below follow the order of the corresponding results in
Sections~\ref{sec:free}, \ref{sec:full}, \ref{sec:parameters}, and
\ref{sec:resonance}. Each example identifies its main results and,
where relevant, additional related statements.

\subsection{Free-group finite-image criteria}\label{sec:free-examples}
\begin{example}[Definite conjugate forms without integrality]\label{ex:nonintegral}
\noindent\textbf{Related results:} Section~\ref{sec:free}, Theorem~\ref{thm:free} and Lemma~\ref{lem:lattice} (the arithmetic hypothesis).

This example shows why definiteness at every coefficient embedding
must be accompanied by an arithmetic assumption to force finite image.
Take $n=2$, $g_1=-1$, $g_2=1$ and $\lambda=(3+4i)/5$.
The input representation of $F_n$ has finite image, but the one-dimensional output has
$T_\lambda(x_1)=-\lambda$ and $T_\lambda(x_2)=1$.
The primitive irreducible polynomial $5X^2-6X+5$ shows that $\lambda$
is not an algebraic integer, hence not a root of unity, so the image is
infinite. For $\lambda=e^{2\pi il}$ with $0<l<1$, the quotient form
in the first coordinate is $4\cos(\pi l)=8/\sqrt5$. At the conjugate
embedding it is $-8/\sqrt5$. Both are definite. Thus definiteness at
all embeddings does not replace the arithmetic assumption in
Theorem~\ref{thm:free}.
\end{example}

\begin{example}[A finite reducible output with indefinite canonical form]\label{ex:reducible-indefinite}
\noindent\textbf{Related results:} Section~\ref{sec:free}, Theorem~\ref{thm:free} (constituentwise definiteness); Section~\ref{sec:full}, Theorem~\ref{thm:unitary-lift} (sign correction).

This example shows why the finite-image criterion tests definiteness
on each convolved constituent rather than on the whole canonical form.
Take $n=2$, $g_1=g_2=g=\operatorname{diag}(i,-i)$, $s_1=I_2$, and
$\lambda=-1$. Then $K=0$ and $L=\{(gv,v):v\in\C^2\}$.
The quotient map $\pi(v_1,v_2)=v_1-gv_2$ identifies all three induced
operators $A_1,A_2,S_1$ with $-g$; the full output is cyclic of order four.
Each quotient class has a unique representative $(w,0)$ with
$w\in\C^2$. When the input form is represented by $I_2$, the canonical quotient
form \eqref{eq:canonical-quotient-form} is represented by the Hermitian
matrix $\operatorname{diag}(2,-2)$ in these coordinates.
Changing the sign on the negative input component gives $2I_2$.
Thus total canonical definiteness is not necessary for finite reducible
output; componentwise definiteness is the correct condition.
\end{example}

\begin{example}[A non-rigid finite input detected as infinite after middle convolution]\label{ex:nonrigid-infinite}
\noindent\textbf{Related results:} Section~\ref{sec:free}, Theorem~\ref{thm:free} and Corollary~\ref{cor:angles} (infinite output without a rigidity assumption).

This example shows how an indefinite canonical form detects infinite
output from a non-rigid finite input, without a braid extension.
On $\C^2$, put
\[
 a=\begin{pmatrix}0&1\\1&0\end{pmatrix},\qquad
 b=\begin{pmatrix}1&0\\-1&-1\end{pmatrix},\qquad
 (g_1,g_2,g_3,g_4)=(a,a,b,b).
\]
The relations $a^2=b^2=(ab)^3=I_2$ give the irreducible two-dimensional
representation of $\mathfrak S_3$, so the input is finite. The positive
form $\left(\begin{smallmatrix}2&1\\1&2\end{smallmatrix}\right)$ is
invariant. Recall that the rigidity index of a product-one tuple
$(h_1,\ldots,h_r)\in\GL_N(\C)^r$ is
$(2-r)N^2+\sum_{j=1}^r\dim Z(h_j)$, where
$Z(h_j)=\{X\in M_N(\C):Xh_j=h_jX\}$ is its matrix centralizer
\cite{Katz}. For irreducible tuples, index $2$ is the rigid case.
Here $P_4=I_2$, so the local monodromy at infinity is trivial and may be
omitted from the tuple. The remaining four matrices have centralizer
dimensions $2,2,2,2$, giving index $(2-4)\cdot4+8=0$; hence the
tuple is not rigid. At $\lambda=-1$, each local spectrum is $\{1,-1\}$ and
the eigenangles of $\lambda P_4$ are $1/2,1/2$. Hence
\[
 R=4,\qquad r=0,\qquad m=2,
\]
and the quotient has signature $(2,2)$. It is irreducible, so its
free-group image is infinite. This is a tuple-level example of the
criterion without rigidity; no braid extension is asserted here.
\end{example}

\subsection{Lifting unitarizability and full-image finiteness}\label{sec:full-examples}

\begin{example}[The lifting theorem with infinite unitary input]\label{ex:infinite-unitary-input}
\noindent\textbf{Related results:} Section~\ref{sec:full}, Theorem~\ref{thm:unitary-lift} (infinite unitary input).

This example illustrates positive unitarizability lifting to the full
output for infinite input, without implying finiteness of that output.
Let $a\in(0,1/3)$ be irrational, put $t=e^{2\pi ia}$, and take
$n=2$, $g_1=g_2=t$, $s_1=1$, $\lambda=t$.
The full input is unitary and infinite. The single input constituent
has $R=2$, $r=0$ and $m=3a-\{3a\}=0$, so its output is positively
unitarizable by Theorem~\ref{thm:unitary-lift}. Its free generator has
eigenvalue $t^2$ of infinite order, so the full output is infinite.
This illustrates the scope of the structural theorem beyond finite inputs.
\end{example}

\begin{example}[A braid permutes distinct input types]\label{ex:braid-permutation}
\noindent\textbf{Related results:} Section~\ref{sec:full}, Theorem~\ref{thm:unitary-lift} and Remark~\ref{rem:braid-components} (permutation of input types).

This example illustrates the permutation of input isotypic components
by the braid action and the compatibility of their output signs.
Take $g_1=\operatorname{diag}(i,1)$,
$g_2=\operatorname{diag}(1,i)$, and
$s_1=\begin{psmallmatrix}0&1\\1&0\end{psmallmatrix}$.
The input has finite image. At $\lambda=-1$, $L=0$ and the quotient
coordinates are the first coordinate of $v_1$ and the second of $v_2$.
The output generators are
\[
 T_{-1}(x_1)=\operatorname{diag}(-i,1),\qquad
 T_{-1}(x_2)=\operatorname{diag}(1,-i),\qquad
 T_{-1}(\sigma_1)=\begin{pmatrix}0&1\\1&0\end{pmatrix}.
\]
The full image has order $32$, namely $(\Z/4\Z\times\Z/4\Z)\rtimes(\Z/2\Z)$.
The two convolved components have the same positive sign and are
exchanged by the braid action.
\end{example}

\begin{example}[Finiteness of $\rho(F_n)$ does not imply finiteness of $T_\lambda(\Gamma)$]\label{ex:full-input-assumption}
\noindent\textbf{Related results:} Section~\ref{sec:full}, Corollary~\ref{thm:full} (the finite full-input hypothesis), compared with Theorem~\ref{thm:unitary-lift}.

This example shows that finite free input and finite free output do not
ensure finite full output without control of the input braid action.
Take $n=2$, $g_1=g_2=i$, $s_1=2$, and $\lambda=-1$. The relations \eqref{eq:artin-action} hold for these matrices, but the full input is infinite and is not unitarizable.
Here $K=0$, $L=\C(i,1)$, and the quotient coordinate $v_1-i v_2$ gives
\[
 T_\lambda(x_1)=T_\lambda(x_2)=-i,
 \qquad T_\lambda(\sigma_1)=-2i.
\]
The free output is finite of order four, while the full output is
infinite. This explains the full-input assumption in
Corollary~\ref{thm:full} and does not contradict
Theorem~\ref{thm:unitary-lift}, whose input must be unitarizable.
Alternatively, replace $s_1=2$ by $s_1=e^{2\pi i\alpha}$ with
$\alpha\in\RR\setminus\Q$. The full input is then unitary, but the output
braid generator $-ie^{2\pi i\alpha}$ still has infinite order.
Thus finite full input in Corollary~\ref{thm:full} cannot be replaced
by positive unitarizability of the full input.
\end{example}

\subsubsection{A scalar specialization with finite full image}\label{sec:scalar-example}
\noindent\textbf{Related results:} Section~\ref{sec:full}, Corollary~\ref{cor:full-spectral}; Section~\ref{sec:resonance}, Proposition~\ref{prop:resonance-splitting} and Theorem~\ref{thm:braid-dichotomy}; Section~\ref{sec:free}, Proposition~\ref{prop:nonroot}.

This example follows the construction from a one-dimensional finite
input to a two-dimensional unitary output with certified finite image.
It also illustrates how the braid image can remain finite for every
parameter while the full image becomes infinite away from roots of
unity. The finite-group calculation also illustrates the splitting at exceptional parameters of Section~\ref{sec:resonance}.

Take $n=3$, $V=\C$, $g_1=g_2=g_3=i$, $s_1=s_2=1$ and $\lambda=i$.
The input is cyclic of order four. For the two embeddings of $\Q(i)$,
the data $(R,r,m)$ are $(3,1,1)$ and $(3,1,3)$, giving signatures
$(2,0)$ and $(0,2)$. Thus Corollary~\ref{cor:full-spectral} certifies
finiteness of the full image.

Here $K=0$ and $L=\C(-1,i,1)^t$. The quotient map
\[
 \pi(v_1,v_2,v_3)=(v_1+v_3,v_2-i v_3)
\]
gives the following induced matrices; within this example we use
$S_i$ and $A_j$ also for the operators on the quotient:
\begin{align*}
 S_1&=\begin{pmatrix}0&i\\1&1-i\end{pmatrix},&
 S_2&=\begin{pmatrix}1&1\\0&-i\end{pmatrix},\\
 A_1&=\begin{pmatrix}-1&i-1\\0&1\end{pmatrix},&
 A_2&=\begin{pmatrix}1&0\\-1-i&-1\end{pmatrix},&
 A_3&=\begin{pmatrix}-i&-1-i\\-1+i&i\end{pmatrix}.
\end{align*}
All preserve the positive form
\[
 H=\begin{pmatrix}2&1-i\\1+i&2\end{pmatrix}.
\]
Let $(u,v)=\pi(v_1,v_2,v_3)$ be the quotient coordinates. Under
$X=u+(1-i)v/2$, $Y=(1+i)v/2$, the form is
$2(|X|^2+|Y|^2)$ and the induced braid matrices $S_1,S_2$ become,
respectively,
\begin{equation}\label{eq:scalar-braid-unitary}
 \widehat S_1=\frac12\begin{pmatrix}1-i&1+i\\1+i&1-i\end{pmatrix},\qquad
 \widehat S_2=\begin{pmatrix}1&0\\0&-i\end{pmatrix}.
\end{equation}
This is a classical finite Burau specialization: the matrices
$C^{-1}S_iC$ for $i=1,2$, where
$C=\left(\begin{smallmatrix}1&i\\-1&1-i\end{smallmatrix}\right)$,
are the reduced Burau matrices in the convention of \cite[Proposition 3.1]{FunarKohno}, where this
braid image has order $96$ and presentation
$\langle u,v\mid uvu=vuv,\ u^4=v^4=1\rangle$.
It is therefore $G_8$ in the Shephard--Todd notation
\cite[Section 2.1]{Chavli}. Moreover, direct
multiplication gives
\[
 A_1=S_1^2S_2^2S_1^2,\qquad A_2=S_1^{-1}S_2^2S_1,
 \qquad A_3=S_1S_2S_1^{-1}S_2S_1^{-1}S_2.
\]
Thus the free matrices already lie in the braid image, and the full
image has the same order.

For the same fixed input, the intermediate braid matrices fix
$v_0=(-1,i,1)^t$ and preserve $\ell=(1,1,1)$, with $\ell v_0=i$.
Thus the generic braid representation splits as
\[
 T_{\rm gen}\simeq\mathbf1\oplus\beta,
 \qquad \beta:B_3\to\GL_2(\C),\quad\beta(\sigma_i)=\widehat S_i\ (i=1,2).
\]
Since $P_3=-i$, the only exceptional parameter is $\lambda=i$. Consequently
$T_i|_{B_3}\simeq\beta$, while
$T_\lambda|_{B_3}\simeq\mathbf1\oplus\beta$ for $\lambda\ne0,1,i$.
The braid image has order $96$ for every parameter. At a parameter which
is not a root of unity, the full image is nevertheless infinite by
Proposition~\ref{prop:nonroot}.

\subsection{Complete finite-image parameter determinations}\label{sec:parameter-examples}

\subsubsection{A two-dimensional noncommutative input}\label{ex:s3-parameters}
\noindent\textbf{Related results:} Section~\ref{sec:parameters}, Theorem~\ref{thm:parameters} and Proposition~\ref{prop:order-bound}; Section~\ref{sec:free}, Corollary~\ref{cor:angles}.

This example shows how the order bound and the signature test give
an exhaustive parameter classification for a noncommutative finite input.
Take $n=2$ and
\[
 a=\begin{pmatrix}-1&1\\0&1\end{pmatrix},\qquad
 b=\begin{pmatrix}1&0\\1&-1\end{pmatrix},\qquad
 g_1=a,\quad g_2=b,\quad s_1=ab.
\]
These matrices give the irreducible standard representation of $\mathfrak S_3$:
$a^2=b^2=(ab)^3=I_2$, their generated group has order six, and they
preserve $\bigl(\begin{smallmatrix}2&-1\\-1&2\end{smallmatrix}\bigr)$.
There is no common invariant line (the common commutant consists of
scalars). Moreover
\[
 s_1as_1^{-1}=b,\qquad s_1bs_1^{-1}=b^{-1}ab,
\]
so this is a finite input of $F_2\rtimes B_2$ in our convention.
The local eigenangles are $(0,1/2)$ for each generator; the product
angles are $1/3,2/3$. Consequently $R=2$, $c=0$, and
\[
 m(l)=\lfloor l+1/3\rfloor+\lfloor l+2/3\rfloor.
\]
The form is positive for $0<l<1/3$, indefinite for $1/3<l<2/3$,
and negative for $2/3<l<1$. At $l=1/3,2/3$ the output is
one-dimensional and definite. Away from those points its dimension
is two. Since all character values are rational, conjugating the input
does not change these spectral data.

With $M=6$ and $\delta=1/3$, the sieve is
$\varphi(q)/2^{\omega(q)}\le6$. It leaves exactly 65 positive integers
$q$, with maximum 420. This count includes $q=1$, which corresponds
to the excluded parameter $\lambda=1$; hence 64 orders remain to be tested.
Here is an explicit finite check using only integer arithmetic.
The prime-power bounds in the proof of
Proposition~\ref{prop:order-bound} imply that every candidate divides
\[
 B_6=2^4 3^3 5^2\cdot7\cdot11\cdot13\cdot17\cdot19\cdot23.
\]
Indeed, $2^e\le24$ and $p^{e-1}(p-1)\le24$ for every odd
prime power $p^e$ dividing a candidate.
For a primitive parameter of order $q$, the embeddings of
$\Q(\lambda)$ run through all angles $k/q$ with $\gcd(k,q)=1$;
the input matrices are rational and remain fixed. By the signatures
above, the parameter passes precisely when no such $k$ satisfies
$q<3k<2q$. The strict inequalities retain the definite exceptional
quotients at $k/q=1/3,2/3$. Thus the following pseudocode reproduces
both the candidate count and the accepted orders; $\varphi$ and
$\omega$ are computed from integer prime factorization.
\begin{quote}\small
\begin{verbatim}
Q := []; accepted := []
for q in positive_divisors(B6), in increasing order:
    if q > 1 and phi(q) <= 6 * 2^omega(q):
        append q to Q
        if no k in {1,...,q-1} satisfies
                gcd(k,q) = 1 and q < 3*k < 2*q:
            append q to accepted
output: |Q| = 64, max(Q) = 420,
        accepted = [3,4,6,10]
\end{verbatim}
\end{quote}
Consequently,
\begin{equation}\label{eq:s3-list}
 \mathcal F(\rho)=
 \{\lambda:\operatorname{ord}(\lambda)\in\{3,4,6,10\}\}.
\end{equation}
The order bounds prove that this is exhaustive, rather than a search
cut off at an experimental maximum. The conclusion holds for the
free and the full KLM images.

The resulting rank-two free images are familiar ones. Away from the
two exceptional parameters $e^{2\pi i/3},e^{4\pi i/3}$, quotient coordinates
$w_1=2v_{11}-v_{12}$ and $w_2=v_{21}-2v_{22}$, where
$v_j=(v_{j1},v_{j2})^{\mathsf T}$ for $j=1,2$, give
\[
 \widetilde A_1=\begin{pmatrix}-\lambda&-1\\0&1\end{pmatrix},\qquad
 \widetilde A_2=\begin{pmatrix}1&0\\-\lambda&-\lambda\end{pmatrix},\qquad
 \widetilde S_1=\begin{pmatrix}0&1\\-1&-1\end{pmatrix}.
\]
Conjugating the first two matrices by $\operatorname{diag}(-1,1)$
gives the standard reduced Burau generators of $B_3$ at parameter
$\lambda$. They are the \emph{free} generators of this KLM output;
the input braid group is $B_2$ and its output generator is instead
$\widetilde S_1$, of order three.

For the nonexceptional finite parameters, put $k=\operatorname{ord}(-\lambda)$.
By \cite[Proposition 3.1]{FunarKohno}, the free image has presentation
\[
 \langle u,v\mid uvu=vuv,\ u^k=v^k=1\rangle,
 \qquad u=\widetilde A_1,\quad v=\widetilde A_2.
\]
In the Shephard--Todd notation for finite complex reflection groups,
these are the following groups; see \cite[Section 2.1]{Chavli} for the
correspondence between the presentations and the group names.
\begin{center}
\begin{tabular}{@{}cccc@{}}
\toprule
$\operatorname{ord}(\lambda)$ & $k$ & Free image & Group order\\
\midrule
$6$ & $3$ & $G_4$ & $24$\\
$4$ & $4$ & $G_8$ & $96$\\
$10$ & $5$ & $G_{16}$ & $600$\\
\bottomrule
\end{tabular}
\end{center}
Parameters of order three give the one-dimensional outputs at exceptional parameters
already described above. Thus this example applies the general
classification to noncommutative input and recovers the classical
finite Burau cases together with their quotients at exceptional parameters.

\subsubsection{A scalar input with a unique finite-image parameter}\label{sec:scalar-parameters}
\noindent\textbf{Related results:} Section~\ref{sec:parameters}, Theorem~\ref{thm:parameters} and Proposition~\ref{prop:order-bound}; Section~\ref{sec:free}, Corollary~\ref{cor:angles}; Section~\ref{sec:resonance}, Theorem~\ref{thm:braid-dichotomy}.

This example shows that a fixed finite input can have just one parameter
with finite full output, while its braid output is finite at every parameter.
Take $n=3$, $g_1=g_2=g_3=i$, and $s_1=s_2=1$. The exponent is $M=4$.
Here $R=3$, $c=0$, and the product eigenangle is $\gamma^{(1)}=3/4$. The form is positive for $0<l<1/4$,
positive on the two-dimensional quotient at $l=1/4$, and indefinite
for $1/4<l<1$. Using this latter interval gives $\delta=3/4$ and
$C=8/3$. There are 19 candidate orders, with maximum 66. This count
also includes $q=1$, so 18 orders remain after excluding $\lambda=1$.
For completeness, the all-embedding test also has a short integer
form. Write $\lambda=e^{2\pi ib/q}$ with $1\le b<q$ and
$\gcd(b,q)=1$, and put $L=\operatorname{lcm}(4,q)$.
Use the coefficient field $\Q(\zeta_L)$; its embeddings are indexed
by $1\le t<L$ with $\gcd(t,L)=1$. Set $r_t=tb\bmod q$ in
$\{1,\ldots,q-1\}$. If $t\equiv1\pmod4$, the input remains $i$
and definiteness is equivalent to $4r_t\le q$.
If $t\equiv3\pmod4$, the input becomes $-i$; complex conjugation
of the preceding signature calculation gives the condition
$4r_t\ge3q$. Equality in either test includes the two-dimensional
definite quotient at the corresponding exceptional parameter.
The prime-power bounds with $C=8/3$ imply that every candidate divides
$B_{8/3}=2^3 3^2\cdot5\cdot7\cdot11=27720$.
Here is the complete check, with the sieve written without fractions:
\begin{quote}\small
\begin{verbatim}
Q := [q > 1 dividing 27720 with 3*phi(q) <= 8*2^omega(q)]
accepted := []
for q in Q and b in {1,...,q-1} with gcd(b,q) = 1:
    L := lcm(4,q)
    accept (q,b) if for every t in {1,...,L-1}
        with gcd(t,L) = 1, r := (t*b) mod q satisfies
        (t mod 4 = 1 and 4*r <= q) or
        (t mod 4 = 3 and 4*r >= 3*q)
output: |Q| = 18, max(Q) = 66, accepted = [(4,1)]
\end{verbatim}
\end{quote}
Thus testing all coefficient embeddings gives exactly
\[
 \mathcal F(\rho)=\{i\}.
\]
Thus the ordinary full image of order 96 at $\lambda=i$ in Appendix~\ref{sec:scalar-example}
is the only finite full specialization of this fixed input. Its braid
image remains finite for every $\lambda\ne0,1$, by Theorem~\ref{thm:braid-dichotomy}.
The full image at $i$ is the complex reflection group $G_8$, as identified
in Appendix~\ref{sec:scalar-example}.

\Needspace{7\baselineskip}
\begin{example}[Complete parameter determination for a rank-three non-rigid input]\label{ex:rank-three-finite}
\noindent\textbf{Related results:} Section~\ref{sec:parameters}, Theorem~\ref{thm:parameters} and Proposition~\ref{prop:order-bound}; Section~\ref{sec:free}, Corollary~\ref{cor:angles}.

This example determines all finite-output parameters for a non-scalar
irreducible input without imposing rigidity. Put $\zeta=e^{2\pi i/3}$,
$\lambda=1+\zeta=e^{\pi i/3}$, and define an $F_3$-input on $\C^3$ by
\[
 g_1=\begin{pmatrix}\zeta&0&0\\0&1&0\\0&0&1\end{pmatrix},\quad
 g_2=\begin{pmatrix}0&1&0\\1&0&0\\0&0&1\end{pmatrix},\quad
 g_3=\begin{pmatrix}1&0&0\\0&0&1\\0&1&0\end{pmatrix}.
\]
These matrices generate the monomial group $G(3,1,3)$ of order
$3^3\cdot6=162$: it consists of all $3\times3$ matrices with exactly
one nonzero entry in each row and column, each such entry belonging
to $\{1,\zeta,\zeta^2\}$. The independent diagonal actions and the transitive
coordinate permutations make its representation irreducible.
The product $P_3=g_1g_2g_3$ satisfies $P_3^3=\zeta I_3$ and has
angles $1/9,4/9,7/9$. The two embeddings of the coefficient field
$E=\Q(\zeta)=\Q(\lambda)$ give the following data for
Corollary~\ref{cor:angles}; local angles list only the nonzero angles
of the three generators, with multiplicity.
\begin{center}
\begin{tabular}{@{}lllll@{}}
\toprule
Embedding & Local angles & Product angles & $(R,r,m)$ & Signature\\
\midrule
Identity & $1/3,1/2,1/2$ & $1/9,4/9,7/9$ & $(3,0,0)$ & $(3,0)$\\
Conjugation & $2/3,1/2,1/2$ & $2/9,5/9,8/9$ & $(3,0,3)$ & $(0,3)$\\
\bottomrule
\end{tabular}
\end{center}
The corresponding parameter angles are $l=1/6$ and $5/6$.
Both forms are definite, so Theorem~\ref{thm:free} proves finiteness.
The coefficient field need not contain all eigenvalues of $P_3$;
the table covers all embeddings of the field of matrix entries.
The input group has exponent $18$, so the cyclotomic model discussed in
Section~\ref{sec:free} would use $\Q(\zeta_{18})=\Q(\zeta_9)$.
Each of the two embeddings of $\Q(\zeta)$ has three extensions to
that larger field. By the field-independence observation following
Theorem~\ref{thm:free}, those six tests repeat the two in the table.
The centralizer dimensions of $(g_1,g_2,g_3,P_3^{-1})$ are $5,5,5,3$,
so its rigidity index is $(2-4)3^2+5+5+5+3=0$.

In fact, for this fixed input the complete parameter set is
\[
 \mathcal F(\rho)=\{e^{\pi i/3}\}.
\]
To prove completeness, now allow $\lambda=e^{2\pi i l}$ to vary.
At the identity embedding, Corollary~\ref{cor:angles} gives
$R=3$ and
\[
 m(l)=\lfloor l+1/9\rfloor+\lfloor l+4/9\rfloor
       +\lfloor l+7/9\rfloor.
\]
The open interval $(2/9,5/9)$ has $r=0$, $m=1$, hence signature $(2,1)$.
Its length is $\delta=1/3$, so Proposition~\ref{prop:order-bound}\textup{(ii)},
with $M=18$, restricts the order $q$ of every finite-output parameter to
\[
 q>1,\qquad \frac{\varphi(q)}{2^{\omega(q)}}\le18.
\]
The prime-power bounds in the proof of that proposition imply that
every candidate divides
\[
 B_{18}=2^6 3^4 5^2 7^2
       \prod_{\substack{11\le p\le73\\p\text{ prime}}}p:
\]
indeed, $2^e\le72$ and $p^{e-1}(p-1)\le72$ for every odd
prime power $p^e$ dividing a candidate.
For $\lambda=e^{2\pi ib/q}$, put $L=\operatorname{lcm}(3,q)$
and use the coefficient field $\Q(\zeta_L)$.
Its embeddings are indexed by $1\le t<L$ with $\gcd(t,L)=1$.
Write $u=tb\bmod q$ in $\{1,\ldots,q-1\}$.
When $t\equiv1\pmod3$, the input is unchanged, and the floor formula
above, including the correction at exceptional parameters, gives
definiteness exactly when $9u\le2q$ or $9u\ge8q$.
When $t\equiv2\pmod3$, complex conjugation gives instead
$9u\le q$ or $9u\ge7q$. The non-strict inequalities retain the
definite two-dimensional exceptional quotients. Thus the complete
test uses only integer arithmetic:
\begin{quote}\small
\begin{verbatim}
Q := [q > 1 dividing B18 with phi(q) <= 18*2^omega(q)]
accepted := []
for q in Q and b in {1,...,q-1} with gcd(b,q) = 1:
    L := lcm(3,q)
    accept (q,b) if for every t in {1,...,L-1}
        with gcd(t,L) = 1, u := (t*b) mod q satisfies
        (t mod 3 = 1 and (9*u <= 2*q or 9*u >= 8*q)) or
        (t mod 3 = 2 and (9*u <= q or 9*u >= 7*q))
output: |Q| = 220, max(Q) = 2730, accepted = [(6,1)]
\end{verbatim}
\end{quote}
The divisor bound and the all-embedding test make this an exhaustive
finite verification of the displayed parameter set.
In particular, the conjugate parameter $e^{-\pi i/3}$ is excluded:
at the identity embedding $l=5/6$ gives $(R,r,m)=(3,0,2)$ and
signature $(1,2)$. Conjugating the parameter alone does not conjugate
the fixed input.

We return to $\lambda=e^{\pi i/3}$ for the explicit output below.
For a direct check, $\dim K=6$ and $L_\lambda=0$. Writing
$v_j=(v_{j1},v_{j2},v_{j3})$ for $1\le j\le3$, the quotient
coordinates $(v_{11},v_{21}-v_{22},v_{32}-v_{33})$ give
\[
 A_1=\begin{pmatrix}-1&-1&0\\0&1&0\\0&0&1\end{pmatrix},\quad
 A_2=\begin{pmatrix}1&0&0\\-\zeta-2&-\zeta-1&1\\0&0&1\end{pmatrix},\quad
 A_3=\begin{pmatrix}1&0&0\\0&1&0\\0&\zeta+1&-\zeta-1\end{pmatrix}.
\]
They preserve the Hermitian matrix
\[
 H=\begin{pmatrix}6&3&0\\3&3&\zeta-1\\0&-\zeta-2&3\end{pmatrix}.
\]
Its leading principal minors are $6,9,9$, so $H$ is positive definite.
With the square-root convention of \eqref{eq:canonical-pencil}, the
canonical quotient matrices in these coordinates are $H/\sqrt3$ at
the identity embedding and $-\overline H/\sqrt3$ at the conjugate
embedding. Equivalently, if $\pi:V^{\oplus3}\to\C^3$ is the quotient
map given by these coordinates, then
$\sqrt3\,\widetilde H_\lambda=\pi^*H\pi$ at the identity embedding.
These identities also verify the signatures in the table directly.
The image lies in the discrete matrices over $\Z[\zeta]$ and in the
compact group preserving $H$, which independently proves finiteness.
The group order can be checked by the following finite enumeration.
Represent $a+b\zeta\in\Z[\zeta]$ by $(a,b)\in\Z^2$, with
\[
 (a,b)(c,d)=(ac-bd,\ ad+bc-bd),
\]
and use ordinary matrix multiplication with these coefficients.
Starting from $E_0=\{I_3\}$, set
$E_{k+1}=E_k\cup\{XA_j:X\in E_k,\ 1\le j\le3\}$.
Exact arithmetic gives $|E_{16}|=1295$, $|E_{17}|=1296$, and
$E_{18}=E_{17}$. This closed set is the generated group: the
generators have finite order, so their inverses are positive powers.
Thus the image has order $1296$, with generator orders $2,3,3$;
each generator is a complex reflection, namely a finite-order
nonidentity transformation fixing a hyperplane pointwise. Moreover,
\[
 \begin{gathered}
 A_1A_2A_1A_2=A_2A_1A_2A_1,\qquad
 A_2A_3A_2=A_3A_2A_3,\qquad A_1A_3=A_3A_1,\\
 A_1^2=A_2^3=A_3^3=I_3.
 \end{gathered}
\]
These are the defining relations of the Shephard--Todd group $G_{26}$
in \cite[Section 4]{MarinG26}, with its generators $(t,s_2,s_1)$
sent to $(A_1,A_2,A_3)$. They give a surjection from $G_{26}$ to the
output group. Since both groups have order $1296$, it is an
isomorphism. This identifies the example as a middle-convolution
passage from $G(3,1,3)$ to $G_{26}$ at $\lambda=-\zeta^2$.
Within $E_{17}$, testing $XA_j=A_jX$ for all $j$ gives the six scalar
matrices $\{\pm\zeta^a I_3:0\le a<3\}$, hence a cyclic center of
order six. Testing $\operatorname{rank}(X-I_3)=1$ gives $9$
reflections of order two and $24$ of order three.
The input generators have orders $3,2,2$, so they cannot be conjugate.
Consequently this input does not extend to the standard
$F_3\rtimes B_3$; the example concerns the free-group criterion alone.
\end{example}

\subsection{Exceptional parameters and projective finiteness}\label{sec:braid-examples}

\begin{example}[A nontrivial kernel at an exceptional parameter]\label{ex:resonance-kernel}
\noindent\textbf{Related results:} Section~\ref{sec:resonance}, Proposition~\ref{prop:image-kernel} (a nontrivial kernel attaining the order bound).

This example realizes the nontrivial-kernel criterion of
Proposition~\ref{prop:image-kernel}: the braid image becomes strictly
smaller at an exceptional parameter, and the kernel attains the
upper bound $|H_\mu|$.
Take $n=2$, $g_1=g_2=i$ and $s_1=i$. The full input is finite.
Since $K=0$, the generic braid operator is
\[
 S_1=i\begin{pmatrix}0&i\\1&1-i\end{pmatrix},
\]
with distinct eigenvalues $i,1$, so $G_{\rm gen}\simeq\Z/4\Z$.
At $\lambda=-1$, the removed line is $L=\C(i,1)$, on which $S_1$
acts by $i$. The quotient coordinate $v_1-i v_2$ has braid action $1$.
Thus $G_{-1}$ is trivial and $\mathcal N_{-1}\simeq H_{-1}\simeq\Z/4\Z$.
The finite kernel can therefore be nontrivial, and the group-order
bound can be attained.
\end{example}

\begin{example}[Infinite braid image for every parameter]\label{ex:always-infinite-braid}
\noindent\textbf{Related results:} Section~\ref{sec:resonance}, Theorem~\ref{thm:braid-dichotomy} (the always-infinite alternative).

This example realizes the empty-set alternative in Theorem~\ref{thm:braid-dichotomy}, even when the full input image is finite.
Take $n=2$, $g_1=g_2=-1$, and $s_1=1$. This scalar input satisfies
the Artin relations and has full image $\{1,-1\}$. Here $K=0$ and
$P_2=1$, so $L_\lambda=0$ for every $\lambda\ne0,1$. The braid output
is therefore generated by
\[
 S_1=\begin{pmatrix}0&-1\\1&2\end{pmatrix}=I_2+N,
 \qquad N=\begin{pmatrix}-1&-1\\1&1\end{pmatrix},\qquad N^2=0\ne N.
\]
For every integer $m\ge0$, $S_1^m=I_2+mN$. These matrices are
pairwise distinct, so $T_\lambda(B_2)$ is infinite for every
$\lambda\ne0,1$. Thus the empty-set alternative in
Theorem~\ref{thm:braid-dichotomy} occurs even for finite full input.
\end{example}

\begin{example}[Why the finite braid-input assumption matters]\label{ex:braid-assumption}
\noindent\textbf{Related results:} Section~\ref{sec:resonance}, Theorem~\ref{thm:braid-dichotomy} (necessity and scope of the finite braid-input hypothesis).

This example shows both that the finite braid-input assumption in
Theorem~\ref{thm:braid-dichotomy} is necessary and that this assumption
alone does not give parameter-independence of projective finiteness.
For $n=2$, take scalar free inputs $g_1=g_2=2$ and braid input
$s_1=c\ne0$. Then $K=0$, $P_2=4$, and
\[
 S_1=c\begin{pmatrix}0&2\\1&-1\end{pmatrix}
\]
has eigenvalues $c,-2c$. At $\lambda=1/4$, the line
$L_\lambda=\C(2,1)$ carries eigenvalue $c$, and the quotient carries
$-2c$. For $c=1/2$, the generic braid image is infinite but the
image at the exceptional parameter has order two. Thus the finite input braid-image
assumption in Theorem~\ref{thm:braid-dichotomy} cannot be omitted.
For $c=1$, the input braid image is trivial, yet the generic projective
image is infinite while the projective image at the exceptional parameter is trivial.
For $c=1$, the ordinary images at both generic and exceptional parameters are infinite. Hence the theorem
cannot be strengthened to projective finiteness under this assumption
alone.
\end{example}

\begin{example}[Why the braid extension matters for projective finiteness]\label{ex:free-projective}
\noindent\textbf{Related results:} Section~\ref{sec:resonance}, Proposition~\ref{p:full} (necessity of the braid extension).

This example shows that the equivalence of ordinary and projective
finiteness in Proposition~\ref{p:full} need not hold for finite input
defined only on the free group.
For an input representation defined only on $F_n$, with
$g_1=-1$, $g_2=1$, the nonzero output is one-dimensional with generator
$-\lambda$; its projective image is trivial even when its ordinary image
is infinite.
\end{example}

\end{document}